\documentclass{elsarticle}

\DeclareMathAlphabet{\mathcal}{OMS}{cmsy}{m}{n}

\usepackage{color}
\usepackage[dvipsnames]{xcolor}
\usepackage{amsmath,amsfonts}
\usepackage{tikz}
\usepackage{enumerate}
\usepackage{mathtools}
\usepackage{mathabx}
\usepackage{url}
\usepackage{comment}

\usepackage{algorithm}
\usepackage[noend]{algpseudocode} 

\usepackage{amsmath}

\usetikzlibrary{shapes.gates.logic.US,trees,positioning,arrows}

\usepackage{tikz}
\usetikzlibrary{shapes.geometric, arrows}

\newcommand{\papertitle}{Automated generation of attack trees with optimal 
	shape and labelling}

\newcommand{\sig}{\ensuremath{\Sigma}}

\newcommand{\treesize}{\ensuremath{size}}

\newcommand{\myarrow}[1]{\stackrel{#1}{\to}}

\newcommand{\powerset}[1]{\mathcal{P}(#1)}

\newcommand{\sosrule}[3]{{[\mathit{#1}]}\frac{#2}{#3}} 

\newcommand{\abasicact}{\ensuremath{b}}
\newcommand{\ftop}{\ensuremath{top}} 

\newcommand{\states}{\ensuremath{\mathcal{S}}}
\newcommand{\alphabet}{\ensuremath{\Lambda}}
\newcommand{\transitions}{\ensuremath{\myarrow{}}}

\newcommand{\startstate}{\ensuremath{s_{0}}}

\newcommand{\finalstates}{\ensuremath{F}}

\newcommand{\ltstuple}{(\states,\alphabet,\transitions,\startstate,\finalstates)}

\newcommand{\astate}{\ensuremath{s}}

\newcommand{\anatree}{\ensuremath{t}}
\newcommand{\rsubtree}{\ensuremath{\preccurlyeq}} 

\newcommand{\topev}{\mathit{top}}

\newcommand{\predset}{\mathbb{P}}

\newcommand{\anlts}{\ensuremath{l}}

\newcommand{\atrace}{\tau}  

\newcommand{\tracesfunc}[1]{\ensuremath{\mathit{paths}(#1)}}  
\newcommand{\attacksfunc}[1]{\ensuremath{\mathit{attacks}(#1)}}  

\newcommand{\attackset}{\ensuremath{\mathcal{A}}}

\newcommand{\anattack}{\ensuremath{a}}

\newcommand{\agoal}{\ensuremath{g}}

\newcommand{\refOp}{\ensuremath{\triangleleft}}

\newcommand{\spSet}{\ensuremath{d}}

\newcommand{\loginuser}{\ensuremath{\mathtt{loggingIn}}} 
\newcommand{\exploiting}{\ensuremath{\mathtt{exploit}}} 
\newcommand{\bruteforce}{\ensuremath{\mathtt{forcePsw}}}
\newcommand{\eavesdrop}{\ensuremath{\mathtt{listen}}} 

\newcommand{\treeop}{\Delta}   
\newcommand{\asem}{S}          

\newcommand{\sat}{\vdash}      

\newcommand{\afact}{\anact}

\newcommand{\setfunction}[1]{\ensuremath{\afact_{pos}}}

\newcommand{\access}{\ensuremath{\mathit{as}}}
\newcommand{\credential}{\ensuremath{\mathit{c}}}
\newcommand{\dologin}{\ensuremath{\mathit{l}}}
\newcommand{\eavesdropuser}{\ensuremath{\mathit{eu}}}
\newcommand{\brute}{\ensuremath{\mathit{b}}}
\newcommand{\exploit}{\ensuremath{\mathit{x}}}
\newcommand{\wait}{\ensuremath{\mathit{w}}}
\newcommand{\eavesdropcon}{\ensuremath{\mathit{ec}}}

\newcommand{\anact}{\ensuremath{b}}

\newcommand{\users}{\ensuremath{U}} 
\newcommand{\auser}{\ensuremath{u}} 

\newcommand{\amachine}{\ensuremath{c}}

\newcommand{\servers}{\ensuremath{S}}
\newcommand{\aserver}{\ensuremath{serv}}
\newcommand{\server}{\ensuremath{s}}
\newcommand{\alice}{\ensuremath{\mathit{Bob}}}
\newcommand{\mallory}{\ensuremath{\mathit{Eve}}}
\newcommand{\passwords}{\ensuremath{\mathit{P}}}
\newcommand{\apassword}{\ensuremath{\mathit{psw}}}
\newcommand{\psw}{\ensuremath{\mathit{pwd}}}

\newcommand{\located}{\ensuremath{\mathtt{loggedin}}}

\newcommand{\knows}{\ensuremath{\mathtt{knows}}}
\newcommand{\stores}{\ensuremath{\mathtt{accepts}}}

\newcommand{\accepts}{\ensuremath{\mathtt{accepts}}}

\newcommand{\reducespacingalg}{\linespread{1.2}\selectfont}

\newcommand{\OR}{\ensuremath{\mathtt{OR}}}
\newcommand{\AND}{\ensuremath{\mathtt{AND}}}
\newcommand{\SAND}{\ensuremath{\mathtt{SAND}}}

\newcommand{\basicact}{\ensuremath{\mathbb{B}}}

\newcommand{\sandtree}[1]{\ensuremath{\mathbb{T}_{#1}}}

\newcommand{\power}[1]{\ensuremath{\mathcal{P}(#1)}}

\newcommand{\sem}[1]{[\hspace{-0.05cm}[#1]\hspace{-0.05cm}]_{\mathcal{S\!P}}}

\newcommand{\spgraphset}{\ensuremath{\mathcal{G}}}

\newcommand{\seqop}{\cdot} 

\newcommand{\genseqop}{\seqop} 
\newcommand{\genparop}{\parop} 

\newcommand{\parop}{\parallel}

\newcommand{\spset}[1]{\ensuremath{\mathbb{G}_{#1}}}

\newcommand{\figurespath}{./fig}

\usepackage{amsthm}
\newtheorem{theorem}{Theorem}[section]
\newtheorem{corollary}{Corollary}[section]

\newtheorem{proposition}{Proposition}[section]
\newtheorem{remark}{Remark}[section]
\newtheorem{definition}{Definition}[section]
\newtheorem{example}{Example}[section]

\usepackage[inline]{enumitem}

\RequirePackage{graphicx}
\RequirePackage{mathptmx}      
\RequirePackage{flushend}
\RequirePackage[colorlinks,citecolor=blue,urlcolor=blue,linkcolor=blue]{hyperref}

\begin{document}
	
	\begin{frontmatter}

	\title{
		\papertitle
		}

		\author[inst1]{Olga Gadyatskaya}
		
		\affiliation[inst1]{organization={Leiden University},
			city={Leiden},
			country={The Netherlands}
		}
		
		\author[inst2]{Sjouke Mauw}

		\affiliation[inst2]{organization={DCS, University of 
		Luxembourg},
			city={Esch-sur-Alzette},
			country={Luxembourg}
		}
		
		\author[inst3]{Rolando Trujillo-Rasua \corref{cor1}\fnref{label2}}
		
		\affiliation[inst3]{organization={Universitat Rovira i 
		Virgili},
			city={Tarragona},
			country={Spain}
		}
		
		\fntext[label2]{Rolando Trujillo-Rasua is funded by a Ramon y Cajal 
		grant from the Spanish Ministry of Science and Innovation and the 
		European Union (REF: RYC2020-028954-I). He is also supported by the 
		project HERMES funded by INCIBE and the European Union 
		NextGenerationEU/PRTR. }
		\cortext[cor1]{Corresponding author}
		
		\author[inst4]{Tim A. C. Willemse}
		
		\affiliation[inst4]{organization={Eindhoven University of 
		Technology},
			city={Eindhoven},
			country={The Netherlands}
		}

		\begin{abstract}

This article addresses the problem of automatically generating attack trees that soundly and clearly describe
the 
ways the system can be attacked.
Soundness means that the attacks displayed 
by the attack tree are indeed attacks in the system; clarity means that the 
tree is efficient in communicating the attack scenario. 
To pursue clarity, we introduce an attack-tree generation 
algorithm that minimises the tree size and the information length of its 
labels without sacrificing correctness. 
We achieve this by i) introducing a system model that allows to reason about attacks and goals in an efficient manner, and ii)  by
establishing a 
connection between the 
problem of factorising algebraic expressions and the problem of minimising the 
tree size. 
To the best of our knowledge, we introduce the first attack-tree generation framework
that optimises the labelling and shape of the generated trees, while 
guaranteeing their soundness with respect to a system 
specification. 

		\end{abstract}

	\end{frontmatter}

	\section{Introduction}
\label{sec:introduction}

Attack trees were introduced by Bruce Schneier \cite{schneier1999attack} in 
1999 as a 
graphical method to evaluate and analyse threats. 
Like mind maps, they are meant to break down a goal into 
smaller, more 
manageable subgoals. 
The root 
node of an attack tree 
denotes the attacker's goal; the root node's children establish how that
goal can be refined into a set of smaller subgoals; and so on.
Because of their potential to help a 
security analyst to 
obtain 
an overview 
of the
system vulnerabilities, facilitate 
communications
across the board, and succinctly store very complex threat
scenarios~\cite{Ford-PoEM-2016}, they have gained widespread 
recognition in 
industry 
and academia alike~\cite{shostack2014threat}. 

There exist two fundamental 
ways of refining a goal: by a 
conjunctive refinement expressing that all subgoals are necessary to meet the 
parent goal, or by a disjunctive refinement expressing that a 
single subgoal is sufficient to meet the 
parent goal. 
Both types of refinement are illustrated in 
Figure~\ref{fig:example}, which displays an attack tree example with the 
following high-level description. The 
goal of the attacker (root node) is to gain unauthorised access to a server. 
To do so, the 
attacker must first get a suitable credential for the server, 
and then use the credential to log in remotely. This is expressed as
a conjunctive 
refinement, denoted graphically by an arrow crossing or connecting the edges 
that link a 
parent node with its children. 
A suitable credential can 
be 
obtained by eavesdropping on communications of an honest user who knows 
the server password. Alternatively, the attacker 
can bruteforce the password on the server or use an exploit to create 
a new password. This constitutes a disjunctive refinement, denoted
graphically by the absence of an arrow.

\begin{figure}[t]
	\centering
	\includegraphics[width=0.80\textwidth]{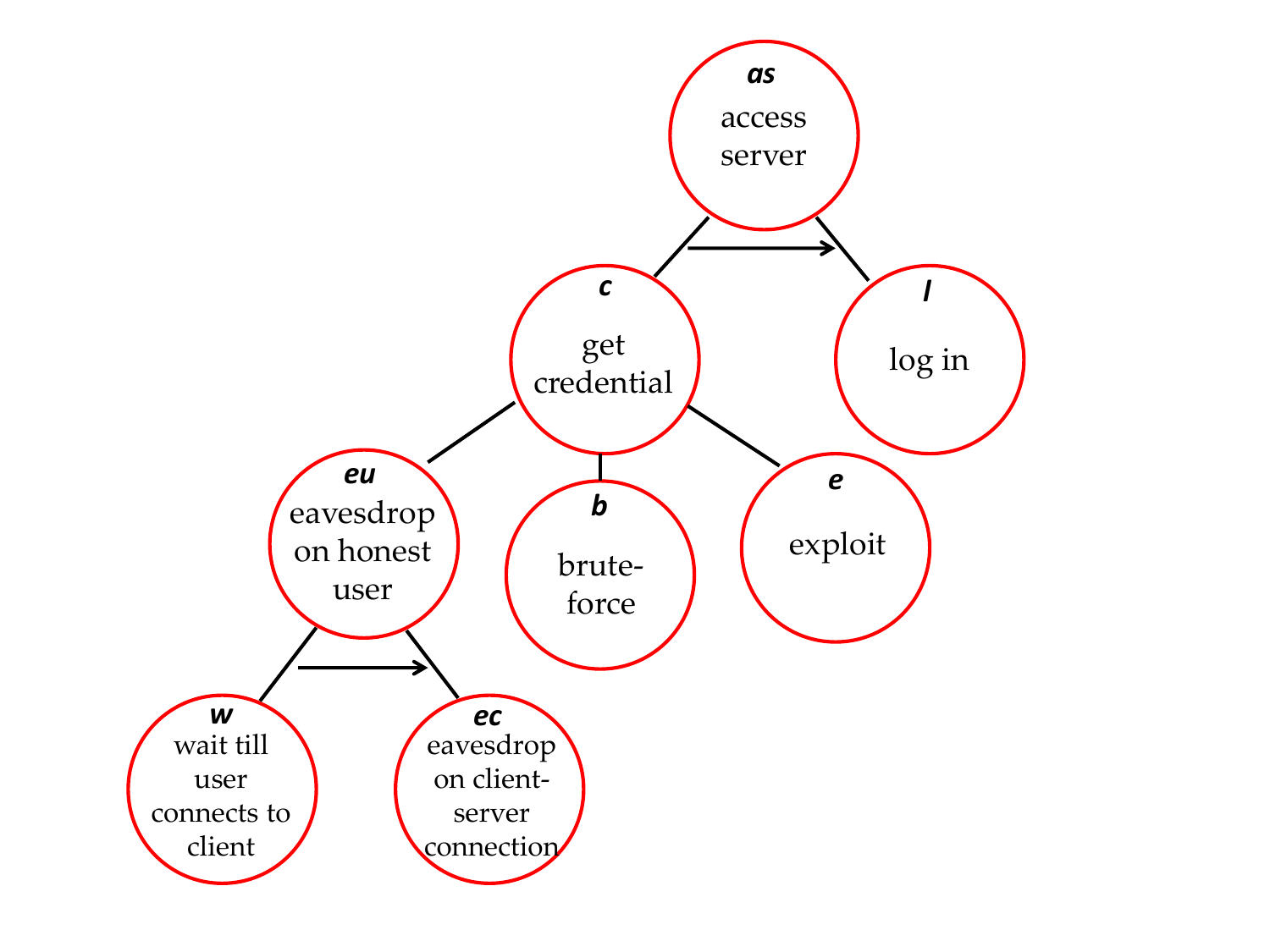}
	
	\caption{An attack tree example. 
	}
	\label{fig:example}
\end{figure}

Notice how the attack tree example has been built over hidden or subjective 
assumptions about the attacker capabilities and the attacked system, making it 
impossible 
to determine how faithfully the attack-tree model captures the most relevant 
attack 
scenarios. In practice, confidence in 
the developed attack-tree models heavily depends on the expertise of the 
security team who builds them. That makes the process of creating an
attack tree time-consuming, resource-intensive, tedious, and
error-prone~\cite{Ford-PoEM-2016,shostack2014threat}. 

Recent academic work 
\cite{Audinot2017,Audinot2018,GJMTW17,Pinchinat2016,Vigo-CSF-2014,Ivanova-GraMSec-2015,
	hong2013scalable,Pinchinat2016,Gadyatskaya-GraMSec-2015,widel2019beyond,Bryans2020,pinchinat2020library,Kern2021}
advocates for 
automated attack-tree 
generation instead of 
building attack trees by hand. 
Their approach, in general,  
consists of three steps:
\begin{enumerate*}[label=(\roman*)]
\item specify the system under analysis,
\item use model checking or other formal methods to find attacks on
  the system, and
\item use automatic means to represent those attacks within an
  attack-tree model. 
\end{enumerate*}
In this case, 
assurances about the completeness and soundness of the attack tree 
come with 
mathematical rigour by resorting to 
state-of-the-art model-checking techniques that    
can systematically and exhaustively explore various combinations of attacks and
system behaviour. The first two steps in this approach benefit from existing
knowledge of system specification and verification. The last step, however, 
that of generating an easy-to-comprehend attack tree out of a system 
specification or set of 
attacks, is still largely an open problem. 

Indeed, because attack trees are meant to be used as part of a threat 
modelling~\cite{shostack2014threat} or risk assessment 
methodology~\cite{paul2014towards}, their communication power (i.e.\ their 
ability to be understood by 
security analysts) is 
as 
important as their soundness and completeness (i.e.\ their ability to faithfully 
capture the most 
relevant threats to the system). 
Notably, attack trees are expected to help the analyst in understanding the
attack 
scenario at different \emph{levels of abstraction} by a simple top down 
inspection of 
the tree. 
In Figure \ref{fig:example}, for example, the first level tells the 
analyst that a credential and a terminal to login is sufficient to access the 
server, but it does not specify how to obtain such a credential. 
The deeper one goes, however, the more details one will find on the attack.

\tikzstyle{arrow} = [thick,->,>=stealth]
\tikzstyle{startstop} = [rectangle, rounded corners, minimum width=2cm, minimum 
height=1cm,text centered, draw=black, fill=green!30]

To make attack trees comprehensible, one has to (at least)
give meaning to their labels. 
Earlier generation approaches that produce trees without labels for 
intermediate-nodes, such as 
\cite{Vigo-CSF-2014}, are clearly unsuitable for a top-down reading.
There exist two complementary approaches in the literature to produce trees 
with 
meaningful labels. The 
first one consists of using 
attack-tree templates \cite{Bryans2020}, 
libraries \cite{Jhawar2018}, or refinement 
specifications \cite{GJMTW17} to guide the generation process. This can help to 
improve both hand-made and computer-generated trees, but it faces the problem 
of creating a comprehensive knowledge base.   
The second approach 
consists of 
providing a formal language to express attack tree goals/labels 
\cite{Mantel2019,Audinot2016,Audinot2017,Audinot2018} and mechanisms to check 
whether 
those goals are consistent with respect to a system specification. 
We deem this necessary,
given the prominent role labels play in the comprehension of an 
attack tree, 
but we have found no 
attack-tree generation framework in the literature capable of producing trees 
with labels expressed within a formal language. 
This article introduces the first such framework.

\noindent \emph{Contributions.}
In this article we introduce an attack-tree generation framework (sketched in 
Figure \ref{fig-our-approach})
that formally establishes the relationship between system specification, 
attacks, goals, and attack trees. 
Like in previous approaches towards automated attack-tree generation, our 
framework starts from a formal specification of the system under analysis and 
relies on formal methods to obtain a (comprehensive) list of attacks against 
the system. Our framework, in addition, requires 
attack
goals to be formally 
derived from the system specification. This allows us to, given an attack-goal 
relation obtained from a system specification, produce attack 
trees with optimised labels and shape.

\begin{figure*}
	\begin{center}
		\begin{tikzpicture}[node distance=3.5cm]
			\node (system) [startstop] {System model};
			\node (attacks) [startstop, right of=system, yshift=1.5cm]{Attacks};
			
			\node (middle) [right of=system]{};
			
			\node (goalsystem) [startstop, right of=system, 
			yshift=-1.5cm]{Goals};
			\node (atree) [startstop, right of=goalsystem, yshift=1.5cm] 
			{Attack 
				Trees};
			\draw [arrow] (system) -- node[anchor=south, sloped] 
			{\small{find}} (attacks);
			\draw [arrow] (system) -- node[anchor=south, sloped] 
			{\small{define}} (middle);
			\draw [arrow] (attacks) -- node[anchor=south, sloped] 
			{\small{heuristic}} (atree);
			\draw [arrow] (system) -- node[anchor=south, sloped] 
			{\small{define}} (goalsystem);
			\draw [arrow, <->] (goalsystem) -- node[anchor=south, sloped] 
			{\small{consistent}} (atree);
			\draw [arrow, <->] (attacks) -- node[anchor=south, sloped] 
			{\small{relation}} (goalsystem);
		\end{tikzpicture}
	\end{center}
	\caption{Our attack tree generation approach, where goals, attacks, and 
		their relation are passed on to the generation algorithm. 
		\label{fig-our-approach}}
\end{figure*}
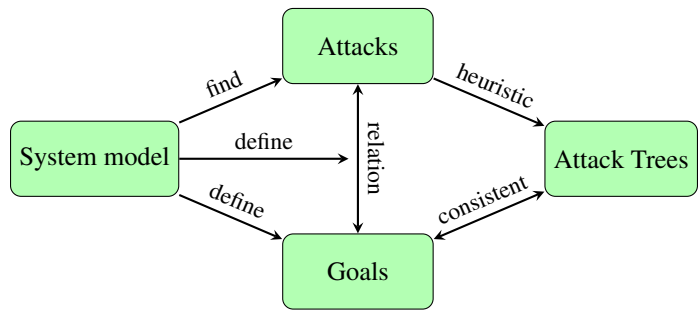

The attack-tree generation framework introduced in this article unifies various 
features that 
appear scattered in the 
literature, namely,

\begin{itemize}
	\item \emph{Flexibility and extensibility 
		\cite{pinchinat2020library,GJMTW17}: }
	By 
	defining an interface between system models and the attack generation 
	problem, our framework allows for 
	plugging in any 
	system model that implements such 
	interface. In our case, the interface considers attacks expressed as 
	series-parallel graphs \cite{JKMR15}, which allows for attack steps to be 
	executed in parallel and sequentially.
	\item \emph{Correctness checking		
		\cite{Mantel2019,Audinot2016,Audinot2017,Audinot2018,
			pinchinat2020library,GJMTW17}: } The framework 
	uses a formal language 
	to produce goals/labels for the tree, which aids their 
	comprehension and gives goals an unambiguous meaning, and establishes 
	whether an attack tree is sound and 
	correctly-labelled with respect to the 
	output of a 
	system specification. 
	\item \emph{Attacks have well-defined goals: }  
	We require attacks to tell which goals they achieve. This 
	has been implicitly assumed in most previous work, but only recently 
	formalised in \cite{Mantel2019}. 
\end{itemize}

In addition, our framework introduces the first \emph{optimality condition} 
for attack tree labels, requiring generation approaches to produce labels that 
convey the minimum necessary information for the analyst to \emph{understand} 
the tree. 
In a nutshell, we define the \emph{meaning} of an attack goal or label as the 
set of 
attacks that achieve the goal. Such meaning is assumed to be prior knowledge 
for the analyst (extracted from the system specification) and thus is independent from the attack-tree generation 
process. 
Then, given an attack tree 
$\anatree$ with root goal 
$\agoal$, we require the meaning of $\agoal$ to be captured by $\anatree$ as 
\emph{tightly} as possible. 
We shall make this more precise in Section 
\ref{sec-generation-problem}.

Another unique characteristic of our framework is the use of graph properties 
as optimisation criteria to shape the generated trees. 
Recognising that the structure of the tree, i.e.\ its nodes and 
connections, impacts the way it is traversed via a top-down reading, we find it 
important that attack-tree generation approaches provide an optimality
criterion for 
their shape. This gives an objective measure to compare different generation 
approaches.  
In our method, we seek 
to produce trees with minimum size. Intuitively, the smaller the tree, the 
fewer 
nodes the 
analyst needs to scrutinise to grasp the attack scenario described by the tree. 

We show instances of the
attack-tree generation 
problem that reduce to the set cover problem. 
To minimise the size of attack trees, we formalise the algebraic properties 
of 
the 
attacks used as input to the attack-tree generation problem, and derive an 
algebraic factorisation procedure to guide the minimisation of the tree size.
This shows
an interesting
connection between the problem of minimising attack 
trees and the problem of minimising algebraic expressions; the latter is often 
done via algebraic factorisation.  

For completeness, we introduce in Section \ref{sec-system} a system model and 
its 
corresponding interface with the attack-tree generation problem. Notably, we 
show how to obtain attack-tree labels out of the system model and how those 
labels relate to attacks. 

All in all, this article introduces the first attack-tree generation method 
that optimises the labelling and shape of the generated trees, while 
guaranteeing their soundness and correctness with respect to a system 
specification.

\noindent \emph{Structured of the article.}
In Section \ref{sec-literature} we discuss the literature on automated attack tree generation and highlight the improvements we make upon previous work. In Section \ref{sec-generation-problem} we define the attack tree generation problem as an optimisation problem where size of the tree and information lost are the variables to be minimised. Section \ref{sec-system} introduces the first component of our attack tree generation framework: a system model from where to extract attacks and their goals in a succint way. Section \ref{sec-factorisation} introduces a heuristic to minimise the attack tree size via leveraging algebraic factorisation techniques.

\section{Related work}
\label{sec-literature}
System-based approaches towards attack-tree generation start from a 
formal specification of the system under analysis, 
typically in the form of a graph describing the system assets, actors, 
access control rules, etc.,  
and produces an attack tree describing how an attacker can lead the system 
into an undesired state. 
That is to say, these approaches translate a specification in a 
system modelling language, such as a labelled transition system, 
into a specification in the attack-tree model. 

Earlier attack-tree generation approaches attempted to make a direct 
model-to-model translation. For example, 
Vigo et al.~\cite{Vigo-CSF-2014} generate trees from a process
calculus system model by translating algebraic specifications into
formulae and backward-chaining these formulae into a formula for the
attacker's goal success. Reachability-based approaches, such as
\cite{Ivanova-GraMSec-2015,Gadyatskaya-GraMSec-2015,hong2013scalable}, 
transform system models
into attack trees using information about connected elements in the
model. In essence, these approaches reason that the attacker can reach
the desired location from any system location adjacent to it. This
reasoning is applied recursively to traverse complete attack paths. Dawkins and 
Hale 
\cite{dawkins2004systematic} have generated attack
trees from network attack graphs (a formalism different from attack
trees \cite{ShHaJhLiWi}) by finding minimum cut sets for successful
attack paths (traces). 
The main drawback of these approaches is that the 
connection between system 
and attack-tree model is made only informally via the generation algorithm.

Later, a series of works on attack-tree generation 
\cite{Audinot2017,Audinot2018,GJMTW17,Pinchinat2016} 
provided simple and elegant soundness and completeness properties for 
the generation algorithm. Given an attack-tree semantics, such as the one 
introduced by Mauw and Oostdijk \cite{MO2006}, an attack-tree generation 
algorithm is sound if all attacks in its semantics are attacks in the system 
specification. In other words, these works use \emph{attack traces} 
as an intermediate representation or interface for both system and attack-tree 
models. Our work follows a similar approach. 
The 
task for the system specification is to find attacks; the task for the 
generation 
algorithm is to represent or synthetise those attacks within an attack-tree 
model.  
Such separation of duties is not only useful to define soundness, but 
also 
to compare and integrate 
different generation approaches. 

In practice, attack trees are seldom used in isolation, though, but as part of 
a 
risk assessment 
methodology. They are meant to succinctly represent an attack scenario, notably 
from the point of view of the attacker, and to aid the comprehension and 
evaluation 
of threats against a system. 
This means that the communication power of an attack tree is as important as 
its soundness. Interestingly, because most attack-tree semantics ignore the 
labels of the intermediate nodes in the tree \cite{Jhawar2016,MO2006},  
soundness is insufficient to guarantee that attack trees are labelled in some 
meaningful way. 

An approach to produce attack trees with meaningful labels is by using 
templates 
or libraries obtained from threat databases 
\cite{Bryans2020,Jhawar2018,GJMTW17,pinchinat2020library,groner2023model}. Both 
the attack tree 
labelling and shape 
are determined, or highly influenced, by a knowledge base on how threats to 
systems are modelled as attack trees. 
For example,  
the ATSyRA approach \cite{Pinchinat-WFMDS-2014,pinchinat2016atsyra}
requires that the
analyst first defines a set of actions at several abstraction levels
in the system model, and a set of rules for refinement of higher-level
actions into combinations of lower-level ones. This action hierarchy
allows to transform successful attack paths into an attack tree, containing 
precise
actions as leaf nodes, while intermediate nodes represent more
abstract actions. This tree enjoys a refinement structure that is more
familiar to the human analyst, but the analyst still has to define the
refinement relation herself. 
In fact, there does 
not exist a comprehensive knowledge base with libraries, templates or 
refinement 
specifications available for these approaches to build attack trees.
Our approach does not need a preestablished library of attack templates to produced labels with meaning, as we extract attack goals and their meaning from the system specification. 

Another line of work starts from the observation that, to obtain meaningful 
labels, one has to formally define what a label is and how it relates to the 
system 
specification. On that direction, 
Audinot et al.\ \cite{Audinot2018} 
establish a robust link between the system 
states and leaf nodes of the tree, with the goal of guiding the analyst on 
refining \emph{useful} tree nodes only. 
Their 
labels are not constrained to be system actions, but reachability
goals defined as words over a set of propositions, which has some similarities 
to our work. 
In \cite{Audinot2017}, labels of their trees are goals 
formulated in terms of 
pre-conditions and post-conditions over the possible states of the system. 
These works, however, do not propose a generation algorithm.

Mantel and Probst \cite{Mantel2019} go one step further. They claim that 
establishing whether an 
attack trace satisfies a given goal is essential to producing trees with 
meaning. We agree with that observation and build our framework around system 
models that can link attack traces and goals precisely. 
Mantel and Probst introduce a formal language for specifying traces and 
attacker goals. They propose different ways to interpret  a 
satisfaction of a goal. For example, a goal may be satisfied by a trace if at 
some point of the trace the goal is satisfied. Or one could use a stronger 
variant: a goal is satisfied if after a certain point of the trace the goal is 
always 
satisfied. Their work focuses on providing a template for defining success 
criteria of 
different flavours, rather than on the generation of attack trees. We, instead, 
focus 
on the attack-tree
generation problem, which can be instantiated with any of the success 
criteria mentioned in \cite{Mantel2019}.

We conclude by highlighting that there also exist approaches to the generation 
of attack trees from informal 
system models, such as \cite{Kern2021}. These approaches are very much in the spirit of the 
attack-tree methodology as suggested by Schneier \cite{schneier1999attack} 
back in 1999. However, as a generation mechanism, they lack formal guarantees of 
correctness. For informal approaches to attack-tree generation we refer the 
reader to \cite{konsta2023survey}; our focus instead is on methodologies whose 
properties can be rigorously analysed. 

\section{The attack-tree generation problem with optimal labelling and 
shape}\label{sec-generation-problem}

The goal of this section is to formalise the attack-tree generation problem as 
an optimisation problem where the attack tree labelling and shape are the 
objective functions. Hence we shall formally introduce the grammar and 
semantics of attack trees, 
their underlying graph structure, their labelling  and, lastly, the attack-tree 
generation problem.

\subsection{$\SAND$ attack trees}
An attack tree defines how higher (\emph{parent}) nodes are
interpreted through lower (\emph{child}) nodes. The interpretations are
defined by the refinement operators: $\OR$ specifies that if any of
the child nodes is achieved, then the immediate parent node is also
achieved; while $\AND$ defines that all child nodes need to be achieved
to achieve the parent node's goal \cite{MaOo}. We will consider also
the sequential $\AND$ operator, or $\SAND$, that demands that the goals
of the child nodes are to be achieved in a sequential order for achieving
the parent node \cite{JKMR15}.

\subsubsection{Syntax.}
Formally, let $\basicact$ denote a set of \emph{node labels}. Let  
$\OR$ and $\AND$ be unranked\footnote{An unranked operator take as argument an arbitrary number of terms.} 
associative and commutative 
operators, and $\SAND$ be an unranked 
associative but non-commutative
operator. 
A \SAND{} attack tree $\anatree$ is an expression over 
the signature $\basicact\cup\{\OR, \AND, \SAND\}$, generated by the 
following grammar (for $\abasicact\in\basicact$):

\[\anatree ::= \abasicact \mid
\abasicact \refOp \OR(\anatree,\ldots,\anatree)
\mid
\abasicact \refOp \AND(\anatree,\ldots,\anatree)
\mid
\abasicact \refOp \SAND(\anatree,\ldots,\anatree)
\text{.}
\]

When drawing $\SAND$ attack trees, like in Figure~\ref{fig:example}, we will 
represent 
the $\AND$ operator by a
bar connecting the edges to the node's children,
the $\SAND$ operator by an arrow
instead of a bar, and the $\OR$ operator by the absence of a
connector. 
Due to the commutativity of the
$\OR$ operator, the order in which an $\OR$ node is connected to its children 
is immaterial;
likewise for the $\AND$ operator. Note that the order of children for $\SAND$ 
nodes remains relevant and is given by the direction of the arrow.

We use 
$\sandtree{\basicact}$ 
to denote all \SAND{} attack trees generated by 
the
grammar above, given a set of labels $\basicact$. Different from the 
definition of \SAND{} trees given
in~\cite{JKMR15}, we require every node in the tree to be annotated
with a label. A label in a node typically provides a generic
(sometimes informal) 
description of the type of attack, e.g. \emph{get
	a user's credentials} or \emph{impersonate a security guard}, which is
helpful for a top-down interpretation of the tree. 
An expression like
$\abasicact \refOp \SAND(\anatree_1,\ldots,\anatree_n)$
denotes an attack tree of which the top node is labelled with 
$\abasicact$, and which has $n$ children $\anatree_1, \ldots,
\anatree_n$ that have to be achieved sequentially. 

\begin{example}\label{example:human}
	Figure~\ref{fig:example} illustrates a $\SAND$ attack tree in which  
	the goal of the attacker is to gain unauthorised access to a server. 
	Using shorthands for labels, this tree can be represented by
	the following expression:

	$
	\access\refOp\SAND(
	\credential\refOp\OR(
	\eavesdropuser\refOp\SAND(
	\wait,
	\eavesdropcon),
	\brute,
	\exploit
	),
	\dologin)
	$.
	
\end{example}

We define the auxiliary function $\ftop$ to obtain the label at the root node
as follows (for $\treeop \in \{\OR, \AND, \SAND\}$):
\[ 
\topev(\abasicact) =
\topev(\abasicact \refOp \treeop(\anatree,\ldots,\anatree)) =
\abasicact\text{.}
\]

Next, we briefly introduce useful notations about the structure of attack 
trees. 
We say that $t'$ is a \emph{subtree} of $t = b \refOp \treeop(t_1, \ldots, 
t_n)$, 
denoted $t' \rsubtree t$, if $t' = t$ or 
$t' \rsubtree t_i$ for some $i \in 
\{1, 
\ldots, n\}$, where $\treeop \in \{\OR, \AND, \SAND\}$. 
If $t$ is a simple tree, i.e.\ a tree with a single node, 
then $t$ itself is a subtree of $t$. 
When the subtree $t'$ is not equal to $t$, we say that $t'$ is a \emph{proper 
	subtree}.
We use the auxiliary function $\treesize(.)$ to denote the size of an attack 
tree 
measured in terms of its number of nodes. A recursive definition of the 
$\treesize$ function is as follows: 
$\treesize(b) = 1$ if $b 
\in \basicact$; $\treesize(\abasicact \refOp \treeop(t_1, \ldots, t_n)) = 
\treesize(t_1) + \dots 
+ \treesize(t_n)+1$ when $\treeop \in \{\OR, \AND, \SAND\}$.

\subsection{The SP semantics}
\label{sec-sp-semantics}
Because attack trees of different shape may well represent the same 
\emph{attack scenario}, researchers have been providing attack trees with
precise and unambiguous interpretations~\cite{Mauw-ICISC-2005,KoMaRaSc,JKMR15}. 
This makes it possible to guarantee consistency of analyses performed 
on attack trees and to determine whether two apparently different trees do 
cover the same attacks. In this article, we use the 
semantics (a.k.a.\ interpretation) of $\SAND$ attack trees
given in \cite{JKMR15}, known as 
\emph{the SP semantics}.

The SP semantics encodes an attack tree as a set of
\emph{Series-Parallel graphs} (SP 
graphs). An SP graph is an edge-labelled directed graph with a \emph{source} 
vertex and a \emph{sink} vertex. The simplest SP graph has the form 
$u \xrightarrow{b} v$, where $b$ is an edge label, $u$ is the \emph{source} 
vertex because it has no incoming edges, and $v$ is the \emph{sink} vertex 
because it has no outgoing edges. Any other SP graph is obtained as the 
composition of single-edge SP graphs. 

Two composition operators are used to build SP graphs: the sequential 
composition operator ($\seqop$) and the parallel composition operator 
($\parop$). A sequential composition joins the sink vertex of a graph with the 
source vertex of the other graph. For example, given $G = u \xrightarrow{b} v$ 
and $G' = x \xrightarrow{z} y$, we obtain that $G \seqop G' = u \xrightarrow{b} 
v \xrightarrow{z} y$. Note that the source vertex of $G'$ has been replaced in 
$G \seqop G'$ by the sink vertex $v$ of $G$. 
A parallel composition, instead, joins the source vertices of both graphs and
joins the sink vertices of both graphs. For example, given $G = u 
\xrightarrow{b} v$ 
and $G' = x \xrightarrow{z} y$, the parallel composition $G \parop G'$ gives 
the 
following SP graph. It is worth remarking that whether vertex $u$ or $x$ is 
preserved in a parallel 
composition is irrelevant as far as the SP semantics is concerned. 

\begin{center}
	\begin{tikzpicture}
		usetikzlibrary{shapes}
		\node (v1) at (0,0) [label=left:$u$] {};
		\node (v2) at (1,0) [label=right:$v$] {};\
		
		\draw [->, bend left=30] (v1.60) to node [above] {$b$} (v2.120);
		\draw [->, bend right=30] (v1.300) to node [below] {$z$} (v2.240);
	\end{tikzpicture}
\end{center}

Both composition operators are extended to sets of SP graphs as follows: 
given sets of SP graphs $\spSet_1, \ldots, \spSet_k$,
\[
\def\arraystretch{1.3}
\begin{array}{l}
	\spSet_1 \genparop \dots \genparop \spSet_k  =
	\{ G_1 \parop  \dots \parop G_k \ \mid \ (G_1, ..., G_k) \in 
	\spSet_1 \times 
	... \times \spSet_k
	\}\\
	\spSet_1 \genseqop \dots \genseqop \spSet_k =
	\{ G_1  \seqop \dots \seqop G_k\ \mid \ (G_1, ..., G_k) \in \spSet_1 \times 
	... \times \spSet_k
	\}.\\
\end{array}
\]

We write $\xrightarrow{b}$ for the graph with a single edge labelled with 
$b$ and define SP graphs as follows. 
\begin{definition}[Series-Parallel graphs] \label{def:sp-graph}
The set $\spset{\basicact}$ of \emph{series-parallel graphs} (SP graphs) over 
$\basicact$ is defined inductively by the following two rules:
\begin{itemize}
\item For $b\in \basicact$, $\xrightarrow{b}$ is an SP graph.
\item If $G$ and $G'$ are SP graphs, then so are 
$G\seqop G'$ and  $G\parop G'$.
\end{itemize}
\end{definition}

\begin{definition}[The SP semantics] \label{def:sp-sem}
	Let $\spset{\basicact}$ denote the set of SP graphs labelled with the 
	elements of  $
	\basicact$. 
	The \emph{SP semantics} for $\SAND$ attack trees is given by the function 
	$\sem{\cdot}: \sandtree{\basicact} \to \power{\spset{\basicact}}$, where $\powerset{\cdot}$ denotes the powerset,
	which is defined recursively as follows:
	for $b\in\basicact$, $t_i \in \sandtree{\abasicact}$, $1\leq i \leq k$, 
	\[
	\def\arraystretch{1.3}
	\begin{array}{l}
		\sem{b} = \{\xrightarrow{b}\}\\
		\sem{\OR(t_1, \dots, t_k)} = \bigcup_{i=1}^{k}{\sem{t_i}}\\
		\sem{\AND(t_1, \dots, t_k)} = \sem{t_1} \genparop ... \genparop 
		\sem{t_k}
		\\
		\sem{\SAND(t_1, \dots, t_k)} = \sem{t_1} \genseqop ... \genseqop 
		\sem{t_k}.\\
	\end{array}
	\]
\end{definition}

We use $\anatree =_{\asem} \anatree'$ to denote semantic equivalence of two 
trees $\anatree, \anatree'\in\sandtree{\basicact}$, which is defined by 
$\anatree =_{\asem} \anatree' \iff \sem{\anatree} = \sem{\anatree'}$.

\begin{example}\label{example:sp-semantics}
	The $\SAND$ attack tree in Figure~\ref{fig:example} has the following SP
	semantics: $\{\xrightarrow{w}\xrightarrow{ec}\xrightarrow{l},
	\xrightarrow{b}\xrightarrow{l}, \xrightarrow{x}\xrightarrow{l}\}$.
	
	Note that, for simplicity, the labels of the graph nodes are not 
	represented in the
	SP semantics. Further note that the SP graphs occurring in this
	example are linear traces because the tree has no $\AND{}$ nodes.
\end{example}

\subsection{The underlying graph structure of attack trees}
\label{sec-shape}

Figure~\ref{fig:flat-tree} shows an attack tree whose semantics is also 
$\{\xrightarrow{w}\xrightarrow{ec}\xrightarrow{l},
\xrightarrow{b}\xrightarrow{l}, \xrightarrow{x}\xrightarrow{l}\}$. That is, 
the attack tree in Figure~\ref{fig:flat-tree} is semantically equivalent to 
the attack tree in Figure~\ref{fig:example}. Their difference lies in their 
underlying graph structure. One tree is flat and does not differ from a mere 
list of 
attacks. The other is deeper and provides potentially useful information in 
the way its nodes are refined, such as the fact that a credential can be 
obtained by three different means: brute-forcing, eavesdropping and 
exploiting a vulnerability. This illustrates that generating 
an attack tree with a given semantics can be trivial. The challenge is to 
generate trees whose refinements contribute to clarity and comprehensibility.

\begin{figure}
	\centering
	\includegraphics[width=1.0\textwidth]{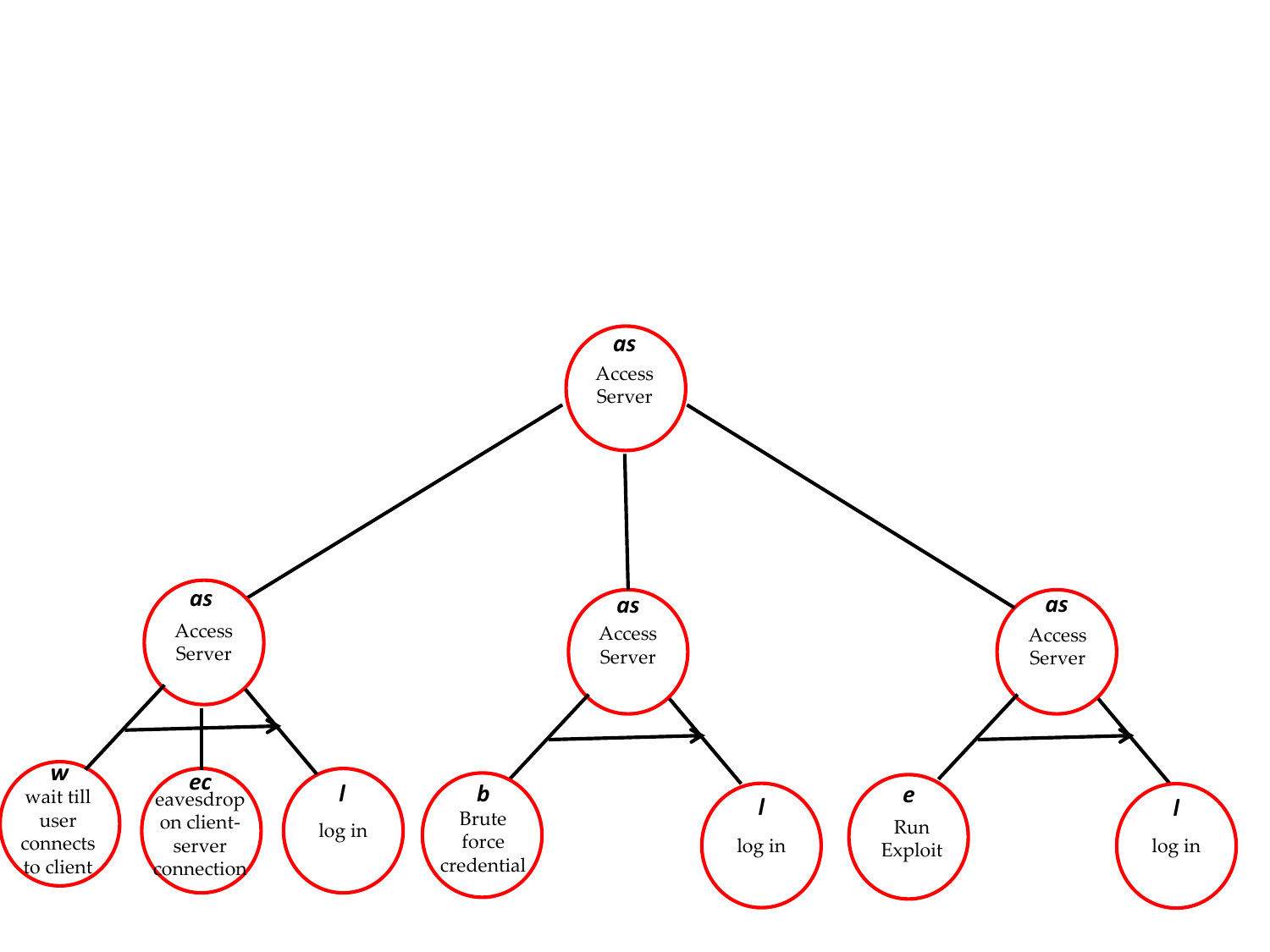}
	\caption{Another way to represent the attack scenario depicted in 
		Figure~\ref{fig:example}. 
		\label{fig:flat-tree}}
\end{figure}	

Existing generation approaches, in general, are not explicit about the graph 
properties of the attack tree they produce. We argue that this is problematic, 
because it thwarts a much needed comparison based on objective measures amongst 
existing 
synthesising 
procedures. While readability and comprehensibility 
might be subjective, graph 
properties are not. In this article we propose a generation approach that seeks 
to minimise the size of the tree, as it represents a simple objective measure 
on the 
number of labels an analyst needs to parse to understand the whole attack 
scenario. 
We remark that if generation approaches to come argue for other structural measures 
for the underlying graph 
structure of the attack tree, such as their depth or branching factor, 
they can be easily accommodated within our 
formulation of the attack-tree generation problem, as we show at the end of this section.

\subsection{The meaning of attacks and goals}
\label{sec-labels}

Like in \cite{Mantel2019}, we assume attacks are not 
arbitrary, or, even if they are, they achieve some goal. Hence we assume the 
existence of a total relation $\sat \subseteq \spset{\basicact} \times 
\basicact$ 
expressing whether an attack satisfies a goal.
The domain of the goal relation follows from the $\SAND$-attack-tree grammar, 
which  
defines an attacker goal as a label in 
$\basicact$, and the SP semantics, which defines an attack as an SP graph over 
$\basicact$. 
To illustrate the concept of the goal relation let us look at the attack tree example in Figure \ref{fig:example}. Consider the attack $\xrightarrow{x}\xrightarrow{ec}$ depicted at the bottom-left of the tree, which means that the attacker waits for the user to connect to the client and then eavesdrops the user's credential. Such an attack accomplishes the goal of learning the user's credential. We express this relationship formally by $\xrightarrow{x}\xrightarrow{ec} \sat c$. 

We observe that existing attack-tree generation procedures have silently 
assumed a 
relation between attacks and goals. As we just illustrated, given a tree 
$\anatree$ with root goal
$\ftop(t)$, it is expected that all attacks captured by $\anatree$ satisfy the 
root goal $\ftop(t)$. 
We make that expectation explicit by characterising when an attack tree is 
\emph{correctly-labelled}. 

\begin{definition}[Correctly-labelled tree]
	We say that an attack tree $t$ is \emph{correctly-labelled}, with respect 
	to a goal relation $\sat$, if every subtree $t'$ 
	in $t$ satisfies 
	that 
	$\anattack \in 
	\sem{t'} \implies \anattack \sat \ftop(t')$.
\end{definition}

Continuing with our running example, if the tree in Figure \ref{fig:example} was correctly-labelled with respect to 
a goal relation $\sat_1$, then it must be the case that $\sat_1$ contains the following relations $(\xrightarrow{w}, w)$, $(\xrightarrow{ec}, ec)$, 
$(\xrightarrow{w}\xrightarrow{ec}, eu)$, $(\xrightarrow{b}, c)$, $(\xrightarrow{x}, c)$, $(\xrightarrow{w}\xrightarrow{ec}, c)$, $(\xrightarrow{w}\xrightarrow{ec}\xrightarrow{l}, a)$, $(\xrightarrow{b}\xrightarrow{l},  a)$, $(\xrightarrow{x}\xrightarrow{l}, a)$, $(\xrightarrow{b}, b)$, $(\xrightarrow{x}, x)$ and $(\xrightarrow{l}, l)$. 
If 
$(\xrightarrow{w}\xrightarrow{ec},  c) \not \in \sat_1$, then the tree in Figure \ref{fig:example} would not be correctly-labelled, because the subtree at the bottom-left with three nodes is not correctly-labelled, but the tree in Figure \ref{fig:flat-tree} would be correctly-labelled if all intermediate nodes are labelled with $a$ (access server).

It is worth remarking that our correctness property on labels is stronger than the one provided in 
\cite{Mantel2019}, as they only require the root node to be correctly labelled. 
We 
believe their choice follows from their attack tree grammar, which lets 
intermediate nodes to go unlabelled. 
We point out, however, that labels are 
as relevant for the root node as for any other node of the tree. Hence we 
require all labels of the tree to be correct. 
Further note that we still have an obligation to show that it is possible to obtain or build a goal relation for security scenarios. We shall address this obligation later in Section \ref{sec-system}, where we introduce a 
system specification language that produces attacks with well-defined goals.

\subsection{The labelling of attack trees}

In the same way as there are many trees with the same semantics but with 
different shapes, there might be many correctly-labelled trees with the same 
semantics but with different labels. It remains to decide 
which, amongst all correctly-labelled trees with a given semantics, 
are better suited to be the output of the attack-tree generation problem. 
We do so by arguing that a root goal of a (sub-)tree 
should give the minimum information necessary for the analyst to make sound 
conclusions about the 
attacks 
the tree is representing. 
Next we define what we mean by \emph{minimum information}. 

Given a goal $\agoal \in \basicact$, let $\attacksfunc{\agoal}$
be the set of attacks satisfying $\agoal$, i.e., 
$\attacksfunc{\agoal} = \{\anattack \in \spset{\basicact}| \anattack \sat 
\agoal\}$. 
Observe that $\attacksfunc{\agoal}$ gives information about $\agoal$
that is independent from 
the attack-tree model. For example, given $c$ and the goal relation $\sat_1$ defined earlier, we obtain that 
$\attacksfunc{c} = \{\xrightarrow{b}, \xrightarrow{x}, \xrightarrow{w}\xrightarrow{ec}\}$. This knowledge is independent from whether 
we choose the attack tree in Figure \ref{fig:example} over the attack tree in Figure \ref{fig:flat-tree}, or vice-versa. 

Now, if we compare the attack tree in Figure \ref{fig:example} with the one in Figure \ref{fig:example-one-attack-missing}, we notice that the latter is missing the attack $\xrightarrow{e} \xrightarrow{l}$. One can notice the difference by inspecting the two trees, but also by comparing the semantics of the trees with respect to the attacks that satisfy the goal $\access$ as established by the goal relation $\sat_1$, which are 
$\attacksfunc{\access} = \{\xrightarrow{w} \xrightarrow{ec} \xrightarrow{l}, \xrightarrow{b} \xrightarrow{l}, \xrightarrow{e} \xrightarrow{l}\}$. Therefore, given a correctly-labelled tree $t$ with root goal $\agoal$, 
we say that an attack $\anattack$ is missing if $\agoal$ 
indicates that $\anattack$ is present, i.e. $\anattack \in 
\attacksfunc{\agoal}$, yet $\anattack$ is not in $t$, i.e. 
$\anattack \not \in \sem{t}$. 
In words, a missed attack is an attack hinting at a label that does not appear 
in the 
tree.

\begin{figure}[t]
	\centering
	\includegraphics[width=0.80\textwidth]{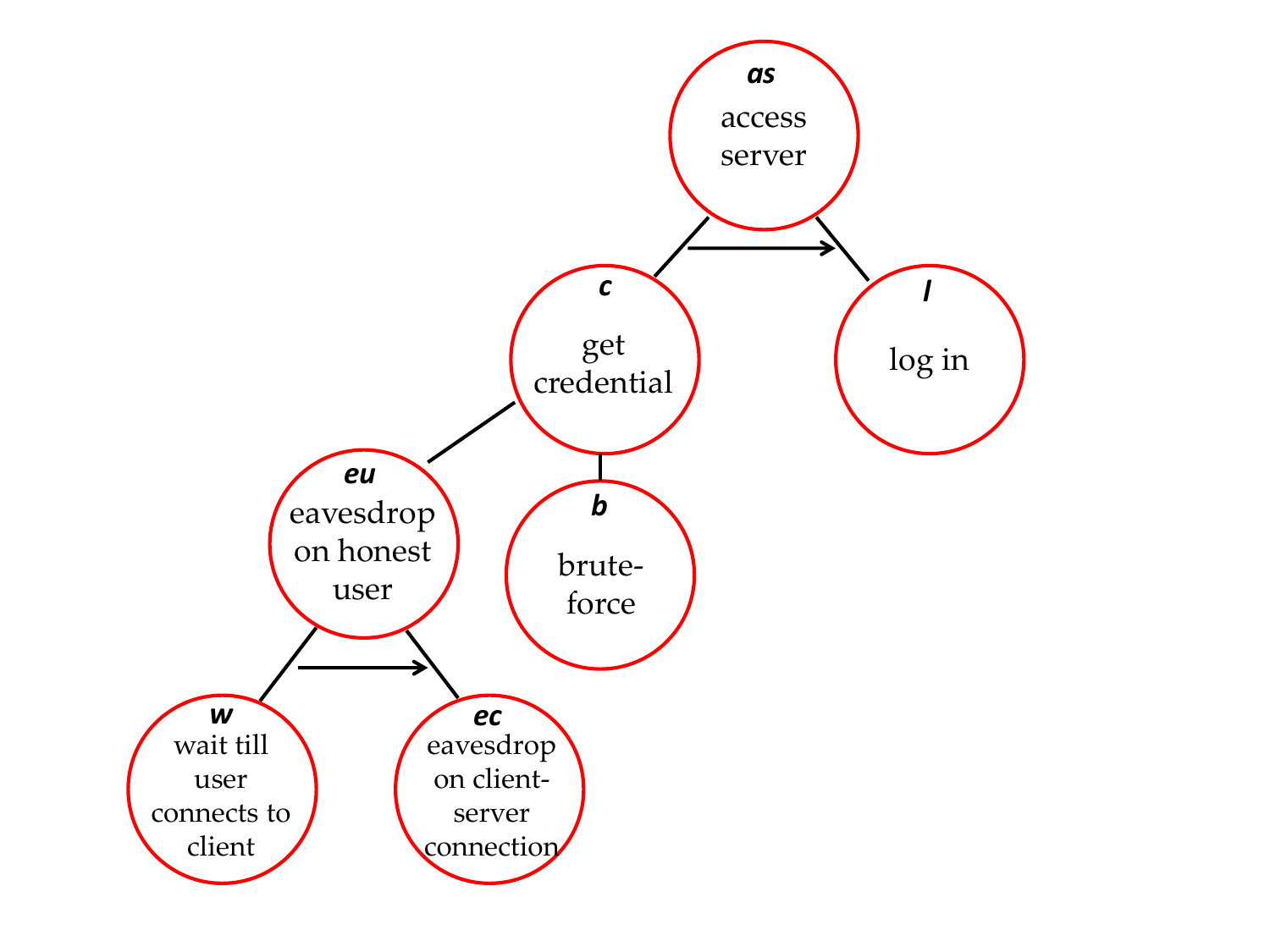}
	\caption{An attack without the attack that uses an exploit. 
	}
	\label{fig:example-one-attack-missing}
\end{figure}


We formalize the notion of missed-attacks as a strict partial order on labels. Given a set of attacks $\attackset$, we say that a label $\agoal$ is \emph{smaller} than another label $\agoal'$ with respect to $\attackset$, denoted $\agoal <_{\attackset} \agoal'$, iff $\attacksfunc{\agoal} \setminus \attackset \subset \attacksfunc{\agoal'} \setminus \attackset$. 
It is easy to prove that $<_{\anatree}$ is a strict partial order. 
This leads to the following definition of optimally-labelled tree: a tree is optimally-labelled if none of its labels can be replaced by a smaller label without making the tree incorrectly-labelled. 


\begin{definition}[Optimally-labelled tree]\label{def-labelling}
	We say that a tree $t$ is \emph{optimally-labelled} if every subtree $t'$ of $t$ satisfies that its root label $\ftop(t')$ cannot be replaced by a smaller label with respect to $<_{\sem{t'}}$ without making $t'$ incorrectly labelled. 
\end{definition}


\subsection{Problem statement}

\begin{definition}[The attack-tree generation problem with optimal labelling 
and shape]
	Given a finite set of attacks $\attackset \subseteq \spset{\basicact} $ and 
	goal 
	relation $\sat \subseteq \spset{\basicact} \times 
\basicact$, \emph{the attack-tree 
		generation problem} consists in finding an 
	optimally-labelled tree $\anatree$ with semantics $\attackset$ and optimal 
	shape (i.e.\ with minimum size). 
	That is, if $\anatree$ is the output of the 
	generation algorithm, then there does not exist another optimally-labelled 
	tree 
	$\anatree'$ 
	with $\sem{\anatree'} = \sem{\anatree} = \attackset$ and $size(\anatree') 
	< size(\anatree)$.

\end{definition}
We finish this section by remarking that other optimality criteria can be 
considered in the attack-tree generation problem, yet insisting that making 
those criteria 
explicit is the only way to objectively compare automated
generation 
approaches.

\section{A heuristic based on factoring to minimise the attack tree size}
\label{sec-factorisation}

In this section we introduce a heuristic to minimise the size of attack trees 
based on factoring algebraic expressions, which has been 
long studied in algebra and digital circuits. 
We start the section by 
motivating the use 
of attack trees in the factored form, then reduce the problem of factorising 
attack 
trees to the problem of factorising algebraic expressions. 

\subsection{Attack trees in the factored form}

Consider the following set of attacks 
$\{\xrightarrow{a}\xrightarrow{b},
\xrightarrow{c}\xrightarrow{d}\}$ and, for the moment,  let's ignore the 
labels of intermediate nodes. 
In this case, it is easy to prove that an attack tree with minimum size with 
such 
semantics is 
$\anatree = \epsilon \refOp \OR(
\epsilon \refOp\SAND(a, b), \epsilon \refOp\SAND(c, d))
$, where $\epsilon$ denotes a placeholder for labels. 
The proof follows from the observation that $\{\xrightarrow{a}\xrightarrow{b},
\xrightarrow{c}\xrightarrow{d}\}$ is \emph{irreducible}, in the sense that it 
cannot be expressed by an attack tree of the form $\abasicact \refOp 
\treeop(\anatree_1, \ldots, \anatree_n)$ with $\treeop \in \{\AND, \SAND\}$. 
We borrow the notion of irreducibility from the area of polynomial 
factorisation where, roughly speaking, 
a polynomial is said to be irreducible if it cannot be expressed as the product 
of two non-constant polynomials. 

\begin{definition}[Reducibility and the factored form for attack trees]
	We say that a set of SP graphs $\spgraphset$ is 
	\emph{reducible} iff there exists a \SAND{} tree of the form $\anatree 
	= \abasicact 
	\refOp 
	\treeop(\anatree_1, \ldots, \anatree_n)$ with $\treeop \in \{\AND, \SAND\}$ 
	such that $\spgraphset = \sem{\anatree}$.
	Moreover, we say that 
	$\anatree$ is a \emph{factorisation} of $\spgraphset$ and that $\anatree$ 
	is in the \emph{factored form} if $\sem{t_1}$, $\ldots$, 
	$\sem{t_n}$ are irreducible.
\end{definition}

Factorisation is a popular technique for the simplification of algebraic 
expressions and we consider it equally useful for the simplification 
(a.k.a.\ minimisation) of attack trees. The reason is that,
in the same way that 
$x * y$ in general has more terms than $x + y$ in their sum-of-product 
representation, where $x$ and $y$ are algebraic expressions, $\abasicact 
\refOp 
\treeop(\anatree, \anatree')$ with $\treeop \in \{\AND, \SAND\}$ in general 
represents a larger set of SP graphs 
than $\abasicact 
\refOp 
\OR(\anatree, \anatree')$. That is to say, while in attack trees conjunctive 
refinements have a multiplying effect in their semantics, disjunctive 
refinements 
do not. 
Our heuristics to solve the optimal attack-tree generation problem thus 
consists of 
creating attack trees in the factored form whenever possible.
Understanding under which conditions a set of SP graphs is 
reducible/irreducible is the first step in this direction, which we pursue 
next. 

\begin{proposition}\label{prop1}
	A set of SP graphs $\spgraphset$ is irreducible if it contains a 
	simple graph\footnote{Recall that a simple SP graph is an SP graph of the form $\xrightarrow{b}.$} or two graphs $g, g' \in \spgraphset$ such that $g$ can be 
	written as 		$g_1 \parop \ldots 
		\parop g_n$ with $n > 1$  
		and $g'$ can be written as $g'_1 
		\seqop \ldots \seqop 
		g'_m$ with $m > 1$. 
\end{proposition} 

\begin{proof}

	The proof of this proposition is simple, yet illustrative. 
	Assume there 
	exists 
	a tree $\anatree $ in factored form whose semantics is $\spgraphset$. 
	This means that $\anatree$ has form $\abasicact 
	\refOp 
	\treeop(\anatree_1, \ldots, \anatree_n)$ with $\treeop \in \{\AND, 
	\SAND\}$. 
	If $\treeop = \AND$ then, according to 
	the SP semantics, $\sem{t} = \sem{t_1} \genparop ... 
	\genparop \sem{t_k}$, in which case neither $g'$ nor any simple graph 
	belong to 
	$\sem{t}$. 
	Otherwise,  
	$\sem{t} = \sem{t_1} \genseqop 
	... \genseqop \sem{t_k}$, in which case neither $g$ nor any simple graph 
	belong to $\sem{t}$.
\end{proof}

The proposition above tells we can only factorise SP graphs when they all are 
in 
sequential composition or they all are in parallel composition, where an SP 
graph is said to be in parallel composition (resp.\ sequential composition) 
if it can be written as $G_1 \parop \ldots \parop G_n$ (resp. $G_1 \seqop 
\ldots 
\seqop G_n$). 
We shall call a set of SP graphs that are all in either sequential or parallel 
composition a \emph{homogeneous} set of SP graphs. Any set with a single 
SP graph is homogeneous too, regardless of whether the graph is simple or a 
composed graph. 
As shown in Proposition \ref{prop1}, being homogeneous is a necessary condition for a set $\spgraphset$ of SP graphs to be reducible, while being reducible is a necessary condition for the existence of an attack tree in factored form whose semantics is $\spgraphset$.  
Hence, the first task of our generation procedure (described in Algorithm 
\ref{alg:gen}) is: given a set of SP graphs $\spgraphset$, partition 
$\spgraphset$ into homogeneous subsets $\spgraphset_1, \ldots, \spgraphset_k$.   

We require each subset $\spgraphset_i$ to contain SP graphs that have a common 
goal with respect to some goal relation $\sat$, 
otherwise there would not exist a correctly labelled tree whose 
semantics is $\spgraphset_i$. Moreover, because we seek trees of minimum size, 
Algorithm 
\ref{alg:gen} aims to find the 
smallest partition possible. 
These two requirements on $\spgraphset_1, \ldots, \spgraphset_k$ make the 
partitioning problem 
NP-hard, as we prove next\footnote{The proof generalises 
the proof provided in our 
	previous 
	work \cite{GJMTW17}, as it imposes no restriction on the attacks, nor 
	on 
	the goal
	relation. 
}.

\begin{theorem}\label{theo-cover-reduction}
	Let $\sat$ be a goal relation on a set of graphs $\spgraphset$ and labels 
	$\basicact$. 
	The set 
	cover problem is polynomial-time reducible to the 
	problem of finding the smallest partition $\{\spgraphset_1, \ldots, 
	\spgraphset_{k}\}$ of $\spgraphset$ such that, for every $i 
	\in 
	\{1, \ldots, 
	k\}$, 
	$\spgraphset_i$ 
	is a homogeneous set of SP graphs and there exists a label 
	$\abasicact \in \basicact$ such that $\forall g \in \spgraphset_i \colon g 
	\sat 
	\abasicact$.
\end{theorem}

\begin{proof}
The proof can be found in the Appendix section. 
\end{proof}

We propose in Algorithm \ref{alg-partition} a 
heuristic to 
partition a 
set of SP graphs based on the standard greedy heuristic used to 
approximate 
the set cover problem. 
Starting from a set $C$ with a single element, the algorithm adds 
elements to $C$ (in no particular order) as long as $C$ remains homogeneous 
and 
all SP graphs in $C$ share a common goal. Once $C$ cannot be increased any 
further, $C$ is considered an element of the optimal partition. This process is 
repeated for the remaining SP graphs until no graph is left out of the 
partition. 
 
\begin{algorithm}
	\caption{A heuristic to solve the partition problem as described in Theorem 
		\ref{theo-cover-reduction}.\label{alg-partition}}
	\reducespacingalg
	\begin{algorithmic}[1]
		\Require A goal relation $\sat$  and a set of SP graphs $\spgraphset$. 
		\Ensure A partition $\{\spgraphset_1, 
		\ldots, 
		\spgraphset_{k}\}$ of $\spgraphset$ 
		such that, for 
		every $i 
		\in 
		\{1, \ldots, 
		k\}$, 
		$\spgraphset_i$ 
		is a homogeneous set of SP graphs and there exists a label 
		$\abasicact$ such that $\forall g \in \spgraphset_i \colon g \sat 
		\abasicact$. 
		The optimisation goal of the algorithm is the minimisation of $k$.	
		
		\Function{Partition}{$\sat, \spgraphset$}
		\State Let $S$ be an empty set

		\While{$\spgraphset$ is non empty}
		\State Let $C = \emptyset$
		\ForAll{$g \in \spgraphset$ such that $C \cup \{g\}$ 
			is a homogeneous set of SP graphs}
		\If{$\exists b \colon g \sat b \wedge \forall g' \in C 
			\colon g' \sat b$} \Comment{Always true when $C = \{g\}$, 
			\phantom{\hspace{7.2cm}} which 
			guarantees termination.}
		\label{step-goal-relation1}
		\State Add $g$ to $C$ and remove $g$ from $\spgraphset$ 
		\EndIf
		\EndFor
		\State Add $C$ to $S$
		\EndWhile
		\State \Return $S$
		\EndFunction
	\end{algorithmic}
\end{algorithm}

Let us illustrate the algoritm above on input the semantics $\{\xrightarrow{w}\xrightarrow{ec}\xrightarrow{l},
\xrightarrow{b}\xrightarrow{l}, \xrightarrow{x}\xrightarrow{l}\}$ of our attack tree example depicted in Figure \ref{fig:example} 
and the goal relation $\sat_1$ introduced in the previous section.  

Before proceeding any further, we give in 
Algorithm \ref{alg:gen} the entry point of our generation 
approach. Its goal is to optimally partition the input set into homogeneous 
sets of SP graphs, relying on specialised procedures for the factorisation 
of homogeneous 
sets of SP graphs. 
The algorithm starts by looking for an optimal label $b$ for a tree whose 
semantics 
is $\spgraphset$. Then it checks whether $\spgraphset$ is 
homogeneous, 
in which case it calls a specialised factorisation algorithm, called 
\textsc{TreeFactorisation} and described in Algorithm \ref{alg:factorisation2} 
further below, that produces optimally-labelled trees in factored form given a 
set of homogeneous 
SP graphs. Otherwise, it addresses 
the 
partition problem described earlier via a 
simple greedy 
heuristics. The output $\{\spgraphset_1, \ldots, 
\spgraphset_{k}\}$ of the partition problem  is used to create a \SAND{} tree 
$\abasicact \refOp \OR(t_1, \ldots, t_k)$ where, for every $i \in \{1, \ldots, 
k\}$, $t_i$ is the tree obtained by recursively calling the algorithm 
$\textsc{TreeGeneration}$ (Algorithm \ref{alg:gen})
on input 
$(\sat,\spgraphset_i)$.

\begin{algorithm}
	\caption{Attack-tree generation via factorisation (Entry 
		point)}\label{alg:gen}
\reducespacingalg
	\begin{algorithmic}[1]
		\Require A goal relation 
		$\sat$ and a non-empty set of SP graphs $\spgraphset$ that share a 
		common goal.

		\Ensure An optimally-labelled tree $\anatree$ with semantics 
		$\spgraphset$ and root goal $\abasicact$. 
		The 
		tree 
		has been 
		optimised for minimum 
		size via factorisation. 
		\Function{TreeGeneration}{$\sat, \spgraphset$}
		\State $b$ is the output of 
		\textsc{FindOptimalLabel} on input $(\sat, \spgraphset)$ 
		\Comment{Defined 
		below}
		\State Let $\spgraphset^{par} = \{g \in \spgraphset | g \textbf{ is not 
			in parallel composition}\}$
		\State Let $\spgraphset^{seq} = \{g \in \spgraphset | g \textbf{ is not
			in sequential composition}\}$
		\If{either $\spgraphset^{par}$ or $\spgraphset^{seq}$ is empty}
		\State \Return the output of \textsc{TreeFactorisation} on 
		input 
		$(\sat, \spgraphset)$ \Comment{This procedure is introduced later at 
		Algorithm \ref{alg:factorisation2}. It aims to produced an 
		optimally-labelled tree in factored form with semantics $\spgraphset$.}
		\label{step-paralallel}

		\Else
		\State $\{\spgraphset_1, 
		\ldots, 
		\spgraphset_{k}\}$ is the output of the
		\textsc{Partition} procedure on input $(\sat, \spgraphset)$

		\label{step-covering}

		\label{step-goal-optimisation}
		\State $\forall i$, $t_i$ denotes the output of 
		\textsc{TreeGeneration} on ($\sat, \spgraphset_i$)
		\State \Return $\abasicact \refOp \OR(t_1, \ldots, t_k)$
		\EndIf
		\EndFunction
		\Function{FindOptimalLabel}{$\sat, \spgraphset$}
			\State \Return $b$ such that $\spgraphset \subseteq attacks(b)$ and no $b'$ smaller than $b$ exists such that $\spgraphset \subseteq attacks(b')$
		\EndFunction

	\end{algorithmic}
\end{algorithm}

\begin{theorem}\label{theo-alg-correct1}
	If the \textsc{TreeFactorisation} function defined at Algorithm 
	\ref{alg:factorisation2} further below always gives an 
	optimally-labelled 
	tree, then 
	Algorithm \ref{alg:gen} always gives an optimally-labelled tree. 
\end{theorem}
\begin{proof}
	First, notice that the 
	root goal of the tree produced by 
	Algorithm \ref{alg:gen} via the procedure $\textsc{FindOptimalLabel}$ 
	satisfies Definition \ref{def-labelling} in a straightforward manner. 
	Second, each call to Algorithm \ref{alg:factorisation2} and 
	Algorithm \ref{alg:gen} satisfies that the input set of SP graphs share a 
	common goal, hence ensuring that an 
	optimally-labelled tree can always be constructed. 
	The recursive nature of the algorithm 
	and the assumption that Algorithm \ref{alg:factorisation2} always gives an 
	optimally-labelled 
	tree, 
	ensures that all subtrees constructed by Algorithm \ref{alg:gen} are 
	optimally labelled.
\end{proof}
\begin{remark}
	For the moment it is unclear how $\textsc{FindOptimalLabel}$ can be 
	implemented in an efficient manner. In Section 
	\ref{sec-system} we show how to do so.

\end{remark}

\subsection{Factorising homogeneous sets of SP graphs via algebraic 
factorisation}
\label{sub-parallel}

The next step is to derive a procedure to 
factorise 
homogeneous sets of SP graphs, which we address here by showing that 
homogeneous sets of SP 
graphs have 
algebraic properties that make them amenable to techniques for algebraic 
factorisation. 

\begin{proposition}
Let $\spset{\basicact}^{seq}$ (resp. $\spset{\basicact}^{par}$) be all 
SP 
graphs that are either simple graphs or are in sequential (resp. parallel) 
composition. 
Then $\spset{\basicact}^{seq}$ is a semigroup under the sequential 
composition operator $\seqop$, and that $\spset{\basicact}^{par}$ is a 
commutative 
semigroup under 
the parallel composition operator $\parop$. 
\end{proposition}
\begin{proof}
	First, observe that for any two elements $g, g' \in 
	\spset{\basicact}^{seq}$ (resp. 
	$g, g' \in 
	\spset{\basicact}^{par}$), it follows that $g \seqop g' \in 
	\spset{\basicact}^{seq}$ (resp. 
	$g \parop g' \in 
	\spset{\basicact}^{par}$). Hence each algebraic structure is a magma. 
	Second, because both $\seqop$ and $\parop$ are associative, then both 
	algebraic 
	structures are semigroups. Commutativity of the structure 
	$(\spset{\basicact}^{par}, \parop)$ follows from the commutativity property 
	of $\parop$. 
\end{proof}

\begin{corollary}
The algebraic 
structure 
$(\powerset{\spset{\basicact}^{seq}}, \seqop, \cup)$ is a semiring without 
multiplicative identity, where $\powerset{\cdot}$ denotes the powerset. 
That is, $\seqop$ and $\cup$ are associative operators, $\cup$ is commutative, 
and $\seqop$ distributes over $\cup$. 
Similarly, $(\powerset{\spset{\basicact}^{par}}, \parop, 
\cup)$ is a commutative semiring without multiplicative identity. 
The only difference between sets in parallel composition and sets in 
sequential composition is that $\parop$ is commutative while $\seqop$ is not. 
\end{corollary}

Even though algebraic factorisation has been extensively studied and 
continues 
to be an 
active research area in algebra and number theory, we are not aware of a 
factorisation procedure for a (commutative) semiring without multiplicative 
identity. On the one 
hand, 
there exist many 
factorisation techniques for 
polynomials over 
finite fields, but they are not idempotent. On the other hand, many works 
have 
addressed 
the minimisation of 
circuit design~\cite{BRSW1987} via factorisation of logic functions, which 
are 
idempotent. However, logic functions 
have a cancelling property not present in 
$(\powerset{\spset{\basicact}^{par}}, 
\seqop, \cup)$, namely $a 
\wedge a = a$ and $(a \wedge b) \vee a = a$, and a commutative property not 
present 
in 
$(\powerset{\spset{\basicact}^{seq}}, \parop, \cup)$. 
This leaves, to the best of our knowledge, the problem of factorising 
expressions in $(\powerset{\spset{\basicact}^{seq}}, 
\seqop, \cup)$ and $(\powerset{\spset{\basicact}^{par}}, 
\parop, \cup)$ open. We start addressing this problem by showing that it 
does not have a unique solution.

	\begin{proposition}
		The factored form of a set of SP graphs is not necessarily unique. 
	\end{proposition}
	\begin{proof}
		Notice that both semirings, namely 
		$(\powerset{\spset{\basicact}^{par}}, 
		\parop, \cup)$ and $(\powerset{\spset{\basicact}^{seq}}, 
		\seqop, \cup)$, are idempotent because their addition 
		operator 
		is 
		the set-union 
		operator, which satisfies that $S \cup S = S$.

		Idempotent semirings have no 
		unique 
		factorisation 
		\cite{agudelo2018polynomials}, hence 
		\SAND{} 
		trees have no unique factorisation either. For example, by mere 
		algebraic 
		manipulation we obtain two different 
		factorisations for the same set of SP graphs 
		$\{a \parop a \parop a,\ a \parop a 
		\parop b,\ a \parop b \parop b,\ b \parop b \parop b\} $\footnote{This 
			example 
			is inspired by a similar example on idempotent semirings given in 
			\cite{agudelo2018polynomials}.}.

		\begin{align*}
			\{a, b\}\parop\{a, b\}\parop\{a, b\} & = \{a \parop a \parop a,\ a 
			\parop a 
			\parop b,\ a \parop b \parop b,\ b \parop b \parop b\} \\
			& = (\{a \parop a\} \parop \{a, b\}) \cup (\{b \parop b\} \parop 
			\{a, 
			b\})\\
			& = (\{a \parop a\} \cup \{b \parop b\}) \parop \{a, b\}
		\end{align*}

	\end{proof}

	\subsection{A factorisation heuristic for expressions in idempotent 
	semirings}
	
	Our goal next is to introduce a factorisation algorithm for 
	expressions in the idempotent semiring $(\powerset{X}, \cdot, \cup)$, where 
	$(X, 
	\cdot)$ is 
	a semigroup and 
	$P \cdot Q = 
	\{p \cdot q | p 
	\in P, q \in Q\}$, which shall be used to factorise SP graphs 
	straightforwardly. For 
	simplicity, given that the reader might be familiar 
	with standard algebraic manipulations, we use $x_1 
	\cdots x_n$ as a shorthand notation for $\{x_1 \cdot 
	\ldots \cdot x_n 
	\} \in \powerset{X}$, $x^n$ instead of $\underbracket{x \ldots 
		x}_n$, and $a + 
	b$ as a shorthand notation 
	for 
	$a \cup b$ with $a, b \in \powerset{X}$. This gives, for example, $a^2 + ab 
	+ 
	b^2$ instead of $\{\{a \cdot a\}, \{a \cdot b\}, \{b \cdot b\}\}$. We 
	remark, 
	nonetheless, that the product operator here should not be assumed to be 
	commutative. For example, $aba$ 
	is not equal to $a^2b$, unless $\cdot$ is commutative.

	The solution we propose is inspired by the factorisation 
	procedure for logic functions introduced in \cite{BRSW1987}. 
	Let 
	us start by introducing the necessary notations. 
	
	\begin{itemize}
		\item A \emph{cube} is a set 
		$\{x_1 \cdot 
		\ldots \cdot x_n 
		\} \in \powerset{X}$ of cardinality one denoting the product $x_1 
		\cdots x_n$. It is the simplest algebraic 
		expression one can form, such as 
		$b^2$ and $ab$ with $a, b \in X$. 
		\item An expression in the Sum of Product form (SoP) is an expression 
		of the 
		form 
		$c_1 + \ldots + c_n$ where $c_1, \ldots, c_n$ are cubes.
		\item The \emph{length} of an expression $f$ is the minimum positive 
		integer $n$ 
		such that there exists an SoP expression $c_1 + \ldots + c_n$ equal 
		to 
		$f$.
		\item $x$ is said to be a left-divisor of $f$ if 
		$y$ is the largest expression, called the quotient, such that $f = x 
		\cdot 
		y + r$ for some $r$ called the remainder of the division. 
		For example, $a$ is a left-divisor of 
		$a^2+ab + b^2$ with quotient $a + b$ and remainder $b^2$. Notice that $b$ is not a left-divisor of $a^2+ab + b^2$, unless $\cdot$ was 
		commutative.
		\item $x$ is said to be a right-divisor of $f$ if 
		$y$ is the largest expression, called the quotient, such that $f = y 
		\cdot 
		x + r$ for some $r$ called the remainder of the division. 
		For example,
		$b$ is 
		a 
		right-divisor of $a^2+ab + b^2$ with quotient $a + b$ and remainder $a^2$. 
		Notice that $a$ is not a right-divisor of $a^2+ab + b^2$, unless $\cdot$ was 
		commutative.

	\end{itemize}

	\begin{proposition}
		Let $f$ be an expression in SoP. Let $R$ be the set of all pairs $(x, 
		y)$ 
		such that $x \cdot 
		y
		= c$ with $c$ a cube in $f$. Let $\pi_{1}(R, z) = \{y | (z, y) \in R\}$ 
		and 
		$\pi_{2}(R, z) = \{x | (x, z) \in R\}$ be 
		the 
		projections of the first and second component of $R$, respectively, 
		onto $z$. 
		If $\{x_1, \ldots, x_n\}$ is a set of maximum cardinality such that 
		$\pi_{1}(R, x_1) \cap \cdots \cap \pi_{1}(R, x_n)$ is non-empty, then 
		$x_1 + 
		\cdots + x_n$ is 
		the quotient of the right-division of $f$ by $\pi_{1}(R, x_1) \cap 
		\cdots \cap 
		\pi_{1}(R, x_n)$. Likewise, if $\{y_1, \ldots, y_n\}$ is a set of 
		maximum 
		cardinality such that 
		$\pi_{2}(R, y_1) \cap \cdots \cap \pi_{2}(R, y_n)$ is non-empty, then 
		$y_1 + 
		\cdots + y_n$ is 
		the quotient of the left-division of $f$ by $\pi_{2}(R, y_1) \cap 
		\cdots \cap 
		\pi_{2}(R, y_n)$.
	\end{proposition}
	\begin{proof}
		We just need to prove one of the implications, since the other one is 
		analogous. 
		Let $y = \pi_{1}(R, x_1) \cap \cdots \cap 
		\pi_{1}(R, x_n)$. On the one hand, 
		by definition of $R$, we obtain that there must exist $r$ such that $f 
		= (x_1 
		+ 
		\cdots + x_n)\cdot y + r$. On the other hand, because $x_1 
		+ 
		\cdots + x_n$ is of maximum length, it is the quotient of the 
		right-division of 
		$f$ 
		by $y$. 

	\end{proof}

	Based on the proposition above, we build a simple heuristic (described in 
	Algorithm \ref{alg:factorisation}) to factorise 
	expressions in 
	$(\powerset{X}, \cdot, 
	\cup)$ by greedily looking for a quotient $x$ and a divisor $y$ such that 
	the 
	product of their length is maximum.
	
	\begin{algorithm}
		\caption{Factorisation algorithm over an idempotent 
			semiring}\label{alg:factorisation}
		\reducespacingalg
		\begin{algorithmic}[1]
			\Require An expression $f$ in SoP over an idempotent semiring 
			$(\powerset{X}, \cdot, \cup)$ where $(X, 
			\cdot)$ is 
			a semigroup and 
			$P \cdot Q = 
			\{p \cdot q | p 
			\in P, q \in Q\}$. 
			\Ensure A factorisation of $f$ of the form $g_1 \cdots g_n + r$ 
			where 
			$g_1, 
			\ldots, g_n$ 
			are all irreducible.
			\Function{ExpFactorisation}{$f$}
			\State Let $R$ be the set of all pairs $(x, y)$ 
			such that $x \cdot 
			y
			= c$ with $c$ a cube in $f$. 
			\If{$R$ is empty}
			\Return $f$
			\EndIf
			\State Let $X$ and $Y$ be an empty sets
			\While{$R$ is non empty}
			\State Let $x_{max}$ such that $\pi_{1}(R, x_{max})$ has maximum 
			cardinality
			\State Let $y_{max}$ such that $\pi_{2}(R, y_{max})$ has maximum 
			cardinality
			\If{$|x_{max}| \geq |y_{max}|$}
			\State Add $x_{max}$ to $X$
			\State $Y = \pi_{1}(R, x_{max})$
			\State Remove all pairs $(x, y)$ in $R$ such that $y \not \in 
			\pi_{1}(R, x_{max})$ or $x = x_{max}$
			\Else
			\State Add $y_{max}$ to $Y$
			\State $X = \pi_{2}(R, y_{max})$
			\State Remove all pairs $(x, y)$ in $R$ such that $x \not \in 
			\pi_{2}(R, y_{max})$ or $y = y_{max}$
			\EndIf
			\EndWhile
			\If{$X = \emptyset$}
			\Return $f$ \Comment{No further factorisation is possible}
			\EndIf
			\State Let $X = \{x_1, \cdots, x_n\}$ and $x = x_1 + \cdots + x_n$
			\State Let $Y = \{y_1, \cdots, y_m\}$ and $y = y_1 + \cdots + y_m$
			\State Let $r$ such that $f = x \cdot y + r$ 
			\State Let $g_1 \cdots g_{m} + r$ and $g'_1 \cdots g'_{m'} + r'$ be 
			the output of the \textsc{ExpFactorisation} procedure (recursive 
			call) 
			on 
			input $x$ and 
			$y$, respectively.  
			\State \Return $x' \cdot y' + r$
			\EndFunction
		\end{algorithmic}
	\end{algorithm}
	
	To illustrate how Algorithm \ref{alg:factorisation} works, let us factorise 
	the 
	expression $a^3 + ba^2  + ab^2 + 
	b^3$ over $(\powerset{X}, \cdot, 
	\cup)$ where $(X, \cdot)$ is a commutative monoid. 
	We obtain that $R = \{(a, 
	a^2), (a,ab), (a,b^2), (a^2,a), (a^2,b), (b,a^2), (b,ab), (b,b^2), (b^2,b), 
	(b^2,a), 
	(ab,a), (ab,b)\}$. Therefore, $\pi_{1}(R, a) = \pi_{2}(R, a) = \{a^2, ab, 
	b^2\}$, $\pi_{1}(R, 
	b) = \pi_{2}(R, 
	b)
	\{a^2, ab, b^2\}$ and 
	$\pi_{1}(R, ab) = \pi_{2}(R, ab) = \{a, b\}$.  
	This means that Algorithm \ref{alg:factorisation} 
	should choose $x_{max}$ to be either $a$ or $b$. Let us assume it chooses 
	$b$. 
	Then $X = \{b\}$, $Y = \{a^2, ab, b^2\}$ and $R$ is updated by removing all 
	pairs whose second 
	element 
	is not 
	$a^2$, $ab$ or $b^2$. Any pair whose first element is $b$ is also removed. 
	Hence, 
	in the 
	second loop cycle $R = \{(a, 
	a^2), (a,ab), (a,b^2)\}$. In this case, $a$ is the only element whose 
	projection 
	$\pi_{1}(R, a) = \pi_{1}(R, a) = \{a^2, ab, b^2\}$ is non-empty. Hence $a$ 
	is 
	added to $X$ and 
	$R$ is updated accordingly. At this stage $R$ is empty and the algorithm 
	has 
	found the factorisation $(a+b)(a^2+ab+b^2)$. The algorithm recursively 
	calls 
	itself to factorise the expressions $a+b$ and $a^2+ab+b^2$ independently, 
	returning $(a+b)^3$.

	To finalise this example, let us now consider that the expression $a^3 
	+ ba^2 + ab^2 + 
	b^3$ is taken over $(\powerset{X}, \cdot, 
	\cup)$ where $(X, \cdot)$ is a (non-commutative) monoid. Notice that this 
	is 
	the same expression we factorised earlier; we are just dropping the 
	commutative property.
	We obtain that $R = \{(a, 
	a^2), (a,b^2),$ $(a^2,a), $ $(b,a^2), $ $(b,b^2), $ $(ba,a), 
	(ab,b), (b^2, b)\}$. Therefore, 
	$\pi_{1}(R, a) = \{a^2, b^2\}$, 
	$\pi_{2}(R, a) = \{a^2, ba,\}$,
	$\pi_{1}(R, b) = 
	\{a^2, b^2\}$,  
	$\pi_{2}(R, b) = 
	\{ab, b^2\}$,  
	$\pi_{1}(R, ba) = \{a\}$, 
	$\pi_{2}(R, ba) = \emptyset$, 
	$\pi_{1}(R, ab) = \{b\}$, 
	$\pi_{2}(R, ab) = \emptyset$, 
	$\pi_{1}(R, a^2) = \{b\}$, 
	$\pi_{2}(R, a^2) = \{a, b\}$, 
	$\pi_{1}(R, b^2) = \{b\}$, 
	and $\pi_{2}(R, b^2) = \{a, b\}$. 
	This means that Algorithm 
	\ref{alg:factorisation} chooses $x_{max}$ to be $a$ or $b$. Say $x_{max} = 
	a$.  
	Then $X = \{a\}$, $Y = \{a^2, b^2\}$ and $R$ is updated by removing all 
	pairs 
	whose second 
	element 
	is not 
	$a^2$ or $b^2$. Pairs whose first element is $a$ are also removed. Hence, 
	in the 
	second loop cycle $R = \{(b,a^2), (b,b^2)\}$. In this case, $b$ is the 
	element with the largest 
	projection 
	$\pi_{1}(R, b) = \{a^2, b^2\}$. Hence $b$ is added to $X$, $Y = \{a^2, 
	b^2\}$,  
	and 
	$R$ is updated accordingly. At this stage $R$ is empty and the algorithm 
	has 
	found the factorisation $(a+b)(a^2+b^2)$.

	\subsection{A factorisation heuristic for sets of homogeneous SP graphs} 
	
	Now we are ready to describe our 
	factorisation procedure for sets of homogeneous SP graphs via a 
	straightforward translation to an algebraic expression over the appropriate 
	semiring (see Algorithm \ref{alg:factorisation2}). The algorithm 
	works as follows. If $\spgraphset$ consists of a single SP graph, then a 
	specialised procedure called \textsc{BuildTree} is used to create a tree 
	whose 
	semantics is $\spgraphset$. This procedure uses an $\AND$ gate whenever 
	there 
	is a parallel composition operator in the SP graph, and a $\SAND$ gate 
	whenever 
	there 
	is a sequential composition operator in the SP graph. If $\spgraphset$ 
	contains 
	various SP graphs all in sequential composition, the algorithm proceeds to 
	create 
	a semiring $R = 
	(\powerset{X}, \seqop, \cup)$ where $X$ contains all SP graphs in parallel
	composition that are part of an SP graph in $\spgraphset$. Likewise, 
	if $\spgraphset$ is in parallel composition the algorithm creates 
	a semiring $R = 
	(\powerset{X}, \parop, \cup)$ that is commutative.
	The structure $R$ is used to factorise the expression $e$, which is the 
	representation of $\spgraphset$ as an SoP expression, by calling Algorithm 
	\ref{alg:factorisation}. If $e = f_1 \cdots f_n$, then the algorithm 
	creates a 
	tree of the form $\abasicact \refOp \Delta(t_1, \ldots, t_n)$ with $\Delta 
	= 
	\AND$ if $\spgraphset$ is in parallel composition, otherwise $\Delta = 
	\SAND$. Each attack tree $t_i$ is obtained by recursively calling 
	$\textsc{TreeFactorisation}$ on the set of SP graphs whose algebraic 
	representation is $f_i$. 
	If $e = f_1 \cdots f_n + r$, then we cannot create a tree in the factored 
	form. 
	Hence, we create a disjunction with two branches. One branch contains a 
	tree 
	whose semantics is $f_1 \cdots f_n$, and the other one a tree whose 
	semantics 
	is 
	$r$. Lastly, if $e = r$, meaning that $e$ consists of cubes with a single 
	element, or some of the sets of SP graphs $\spgraphset_1, \ldots, 
	\spgraphset_{n+1}$ has no common goal, 
	then we create a disjunction of depth 1 where each branch of the root node 
	is created by calling $\textsc{BuildTree}$ on input an SP graph  in 
	$\spgraphset$.
	 		
	\begin{algorithm}
		
		\caption{Factorisation algorithm for a set of SP graphs 
			of the same 
			type}\label{alg:factorisation2}
		
		\linespread{1.1}\selectfont
		
		\begin{algorithmic}[1]
			\Require A goal relation 
			$\sat$ and a non-empty set of homogeneous SP graphs that share a 
			common goal.

			\Ensure An optimally-labelled tree $\anatree$ with semantics 
			$\spgraphset$. The produced tree has 
			been 
			optimised for minimum 
			size via factorisation. 
			\Function{TreeFactorisation}{$\sat, \spgraphset$}
			\State Let $b$ be the output of \textsc{FindOptimalLabel} on input 
			$\spgraphset$
			\If{$\spgraphset = \{g\}$}
			\State \Return The output of the \textsc{BuildTree} procedure below 
			on input $(\sat, g)$ 
			\EndIf
			\State Let $X$ be an empty set
			\If{$\spgraphset$ is in parallel composition}
			\ForAll{$g_1 \parop \ldots \parop g_n \in \spgraphset$ where each 
				$g_i$ is not in parallel composition}
			\State Add $g_1, \ldots, g_n$ to $X$
			\EndFor
			\State Let $R = (\powerset{X}, \parop, \cup)$ be a
			semiring with $(X, \parop)$ a commutative semigroup.
			\Else
			\ForAll{$g_1 \seqop \ldots \seqop g_n \in \spgraphset$ where each 
				$g_i$ is not in sequential composition}
			\State Add $g_1, \ldots, g_n$ to $X$
			\EndFor
			\State Let $R = (\powerset{X}, \seqop, \cup)$ be a 
			semiring with $(X, \seqop)$ a semigroup.
			\EndIf

			\State Let $e$ be the $R$-algebraic expression in SoP obtained from 
			$\spgraphset$
			\State Call Algorithm 
			\ref{alg:factorisation} (\textsc{ExpFactorisation}) to factorise 
			$e$ 
			into $f_1 \cdots f_n + r$.
			\State Let $\spgraphset_i$ be a set of SP graphs in $R$ whose 
			algebraic representation is $f_i$
			\State Let $\spgraphset_{n+1}$ be a set of SP graphs in $R$ whose 
			algebraic representation is $r$
			\If{$n = 0$ or $\exists \spgraphset_i$ s.t. its SP graphs have no 
				common 
				goal}\label{step-no-common-goal}
			\State Let $\{g_1, \ldots, g_m\} = \spgraphset$

			\State $\forall g_i$, $t_i$ is the output of 
			\textsc{BuildTree} on input $(\sat, 
			g_i)$
			\State \Return $b \refOp \OR(t_1, \ldots, t_m)$ \Comment{A flat 
				tree is built}	 
			
			\ElsIf{$e$ is factorised as $f_1 \cdots f_n$ with $n>1$} 
			\Comment{$r$ 
				is 
				the empty expression}

			\label{step-goal-optimisation4} 
			\State $\forall \spgraphset_i$, $t_i$ is the output of 
			\textsc{TreeGeneration} (Algorithm \ref{alg:gen}) on input $(\sat, 
			\spgraphset_i)$
			\If{$\spgraphset$ is in parallel composition}
			\Return $b \refOp \AND(t_1, \ldots, t_n)$
			\Else
			\ \Return $b \refOp \SAND(t_1, \ldots, t_n)$
			\EndIf
			\ElsIf{$e$ is factorised as $f_1 \cdots f_n + r$ with $n>1$} 
			\Comment{A 
				partial 
				factorisation
			}

			\label{step-goal-optimisation6} 
			\State $t_l$ is the output of \textsc{TreeGeneration} on input 
			$(\sat, 
			\spgraphset_1 \cup \cdots \cup \spgraphset_n)$
			\State $t_r$ is the output 
			\textsc{TreeGeneration} on input 
			$(\sat, 
			\spgraphset_{n+1})$
			\If {$t_r \equiv b_r \refOp \OR(t_1, \ldots, t_m)$}
			\Return $b \refOp \OR(t_l, t_1, \ldots, t_m)$ 
			\Else 
			\ \Return $b \refOp \OR(t_l, t_r)$
			\EndIf
			
			\EndIf
			\EndFunction
		\end{algorithmic}
	\end{algorithm}
	
	\begin{algorithm}
		\caption{An algorithm that builds an optimally-labelled tree from a 
			single SP graph.}\label{alg:factorisation3}
		\reducespacingalg
		\begin{algorithmic}[1]
			\Require A  goal relation 
			$\sat$ and an SP graph $g$.
			\Ensure An optimally-labelled tree $\anatree$ with 
			semantics $\{g\}$. 
			\Function{BuildTree}{$\sat, g$}
			\State Let $b$ be the output of 
			\textsc{FindOptimalLabel} on input 
			$\{g\}$ 
			\If{$g$ is a simple graph}
			\Return the single-node \SAND{} tree $b$ \Comment{Here $g \equiv b$}
			\EndIf
			\If{$g$ is in parallel composition}
			\State Let $g = g_1 \parop \ldots \parop g_n$ with $n$ maximum
			\Else
			\State Let $g = g_1 \seqop \ldots \seqop g_n$ with $n$ maximum
			\EndIf

			\State $\forall g_i$, $t_i$ is the output of \textsc{BuildTree} on 
			input $(\sat, g_i)$
			\If{$g$ is in parallel composition}
			\Return $b \refOp \AND(t_1, \ldots, t_n)$
			\Else
			\ \Return $b \refOp \SAND(t_1, \ldots, t_n)$
			\EndIf
			\EndFunction
		\end{algorithmic}
	\end{algorithm}
	
	This algorithm is rather cumbersome, hence let us run it by using the 
	semantics 
	of the attack tree displayed in Figure~\ref{fig:example}. 
	The syntax of such a tree 
	is 
	$
	\access\refOp\SAND(
	\credential\refOp\OR(
	\eavesdropuser\refOp\SAND(
	\wait,
	\eavesdropcon),
	\brute,
	\exploit
	),
	\dologin)
	$, where the symbols used are shorthand notations that can be found in the 
	figure itself. 
	It follows that the semantics of such tree is 	
	$\{\xrightarrow{w}\xrightarrow{ec}\xrightarrow{l},
	\xrightarrow{b}\xrightarrow{l}, \xrightarrow{x}\xrightarrow{l}\}$.
	Now, to run the generation algorithm we need a relation between 
	attacks and labels. Let us create the relation that trivially follows from 
	the 
	tree itself. 
	\begin{itemize}
		\item $\xrightarrow{\ell} \sat \ell$ for all labels $\ell$, i.e. the 
		execution of 
		$\ell$ indeed achieves $\ell$
		\item $\xrightarrow{w} \xrightarrow{ec} \sat 
		\eavesdropuser$ 
		and 
		$\xrightarrow{w} \xrightarrow{ec} \sat \credential$, 
		because waiting until the user connects to later eavesdrop his 
		credential 
		achieves the goal 
		of eavesdropping and the goal of obtaining the user's credential. 
		\item $\xrightarrow{b} \sat \credential$ and $\xrightarrow{x} 
		\sat \credential$, because bruteforce and exploitation are independent 
		ways to obtain a credential 
		\item $\xrightarrow{w} \xrightarrow{ec} 
		\xrightarrow{l} \sat \access$
		\item $\xrightarrow{b} 
		\xrightarrow{l} \sat \access$
		\item $\xrightarrow{x} 
		\xrightarrow{l} \sat \access$
	\end{itemize}
	
	The last three relations are inferred from the attack-tree semantics and 
	the 
	root goal $\access$. Now, let us run Algorithm 
	\ref{alg:factorisation2}  on 
	input $(\sat, \{\xrightarrow{w}\xrightarrow{ec}\xrightarrow{l},
	\xrightarrow{b}\xrightarrow{l}, \xrightarrow{x}\xrightarrow{l}\})$, which 
	is a homogeneous set of SP graphs with a common goal $\access$.

	Because $\spgraphset$ is in sequential composition, Algorithm 
	\ref{alg:factorisation2} creates the set $X = \{\xrightarrow{w}, 
	\xrightarrow{ec}, \xrightarrow{l},
	\xrightarrow{b}, \xrightarrow{x}\}$ and the algebraic structure 
	$(\powerset{X}, 
	\seqop, \cup)$. From $\spgraphset$, it creates the expression $e = 
	\xrightarrow{w}\xrightarrow{ec}\xrightarrow{l} + 
	\xrightarrow{b}\xrightarrow{l} + \xrightarrow{x}\xrightarrow{l}$. By 
	calling 
	Algorithm \ref{alg:factorisation}, $e$ is factorised as 
	$(\xrightarrow{w}\xrightarrow{ec} + 
	\xrightarrow{b} + \xrightarrow{x})\xrightarrow{l}$. 
	Next, the algorithm looks for an optimal label
	common to the attacks $\xrightarrow{w}\xrightarrow{ec}, 
	\xrightarrow{b}, \xrightarrow{x}$. Such optimal label is $\credential$, and 
	it 
	is unique. Therefore, the algorithm calls recursively itself to create a 
	tree 
	with root label $\credential$ and semantics 
	$\{\xrightarrow{w}\xrightarrow{ec}, 
	\xrightarrow{b}, \xrightarrow{x}\}$. Let us call such tree $t_1$. Likewise, 
	the 
	optimal label for $\{\xrightarrow{l}\}$ is found to be $l$. The attack tree 
	with root label $l$ and semantics  $\{\xrightarrow{l}\}$ is just the simple 
	tree $l$. Therefore, the algorithm returns the tree $\access \refOp 
	\SAND(t_1, 
	l)$. 
	
	Now, let us see how Algorithm \ref{alg:factorisation2} handles the input 
	$(\sat, \{\xrightarrow{w}\xrightarrow{ec}, 
	\xrightarrow{b}, \xrightarrow{x}\}$. 
	In this case, $X = \{\xrightarrow{w}, \xrightarrow{ec}, 
	\xrightarrow{b}, \xrightarrow{x}\}$ and $e = \xrightarrow{w}
	\xrightarrow{ec} + 
	\xrightarrow{b} + \xrightarrow{x}$. The factorisation of $e$ gives $f_1 = 
	\xrightarrow{w}$, $f_2 = \xrightarrow{ec}$ and $r = \xrightarrow{b} + 
	\xrightarrow{x}$. Therefore, we are in the case where only a partial 
	factorisation was found. 
	The \SAND{} tree corresponding to the remainder 
	$\xrightarrow{b} + 
	\xrightarrow{x}$ is $c \refOp \OR(b, x)$, because $c$ is the shared goal of 
	both attack steps. Likewise, the \SAND{} tree for 
	$\xrightarrow{w}\xrightarrow{ec}$ is $c \refOp \SAND(w, ec)$. This means 
	that 
	the output of the algorithm on input 
	$(\sat, \{\xrightarrow{w}\xrightarrow{ec}, 
	\xrightarrow{b}, \xrightarrow{x}\}$ is the \SAND{} tree 
	$c \refOp \OR(c \refOp \SAND(w, ec), b, x)$, which coincides with the 
	attack 
	tree in Figure \ref{fig:example}.   
	
		\begin{theorem}
		Algorithm \ref{alg:factorisation2} produces optimally-labelled trees. 
	\end{theorem}
	\begin{proof}
		As in the proof of Theorem \ref{theo-alg-correct1}, all we need to 
		show is that a tree can always be produced. The fact that the produced 
		tree is optimally labelled comes from the use of the 
		\textsc{FindOptimalLabel} procedure to create labels, which was already 
		shown to use optimal labels. Now, an optimally-labelled tree exists 
		whenever the set of SP graphs has a common goal. 
		In Step \ref{step-no-common-goal}, 
		we ensure that, if any of the sets $\spgraphset_1, \ldots, 
		\spgraphset_{n+1}$ has no common goal or if no factorisation 
		exists, 
		then our algorithm builds a flat tree where each branch of the root is 
		a tree $t_i$ with semantics $\{t_i\}$ where $\spgraphset = \{g_1, 
		\ldots, g_m\}$. 
		Because the goal relation satisfies that every SP graph $g$ satisfies 
		some goal and, by the premise of Algorithm \ref{alg:factorisation2}, 
		the SP graphs in $\spgraphset$ have a common goal, then we conclude 
		that Algorithm \ref{alg:factorisation2} always produce a tree, and that 
		such a tree is optimally labelled.

	\end{proof}
	
	\begin{remark}
		If the semiring is commutative, then the set $R$ in Algorithm 
		\ref{alg:factorisation} has exponential size in terms of the number of 
		cubes in 
		the expression. The size of each cube also contributes exponentially to 
		the 
		size of $R$. A workaround to this problem is assuming the semiring to 
		be 
		non-commutative based, for example, on a lexicographical order on $X$.
	\end{remark}

	\section{A system model to efficiently reason about attacks and goals}
	\label{sec-system}
	
	The input to the attack tree generation problem is a set of attacks and a goal relation. Previous work have used system models that can produce set of attacks, such as 
	\cite{Audinot2017,Audinot2018,GJMTW17,Pinchinat2016}, while others have focused on labels, such as 
	\cite{Bryans2020,Jhawar2018,GJMTW17,pinchinat2020library,groner2023model,Pinchinat-WFMDS-2014,pinchinat2016atsyra,Audinot2018,Audinot2017,Mantel2019}. But, to our knowledge, none can be used to efficiently find the optimal labelling of an attack tree with a given semantics. 
	The goal of this section is to introduce a 
	system model 
	capable of producing attacks, a non-trivial goal relation, and a procedure to efficiently find an optimal labelling for an attack tree.


	\subsection{Mixed Kripke Structure: the system model}
	
	A Kripke structure is a variation of the \emph{Labelled
		Transition Systems} (LTS), which are used to describe 
	the behaviour
	of a system by defining the transitions that bring a system from one
	state into another. Each state in a Kripke structure contains properties that 
	hold in the corresponding state, 
	such 
	as the knowledge of an attacker or the location of a secret document. 
	The transition system 
	progresses via application of transition rules until a final state is reached. 
	
	Formally, a \emph{Mixed kripke Structure} is a quadruple 
	$\ltstuple$, where $\states$ is a set of \emph{states}; $\alphabet$ is a set of 
	\emph{labels}; $\transitions \subseteq \states\times\alphabet\times\states$ is a
	\emph{transition relation}; 
	$\startstate \in\states$ is 
	the \emph{initial 
		state} and $\finalstates \subset \states$ is a set of final states. 
	We usually write $s \xrightarrow{a} s'$ to say that $(s, a, s') 
	\in \transitions$. And we say that 
	$s_0 
	\xrightarrow{a_0} s_1 \xrightarrow{a_1} s_2 \cdots s_{n-1} 
	\xrightarrow{a_{n-1}} s_n$ 
	is a path if $s_n \in \finalstates$  and 
	$s_i \xrightarrow{a_i} s_{i+1}$ for 
	every 
	$i 
	\in \{0, \ldots, n-1\}$. 
	Any portion of the path formed by one or more 
	consecutive transitions is called a \emph{subpath}. 
	When $s_0$ is a final state, then $s_0$ itself is 
	a path with no transition.
	We use $\tracesfunc{\anlts}$ to denote the set of all 
	paths
	generated by a labelled transition system $\anlts$.

	\paragraph{States} As mentioned earlier, states contain system properties, 
	which we formalise as predicates constructed over a signature
	$\sig$ (i.e.\ a 
	collection of function symbols). Function symbols of arity zero are called 
	constants and are written without parentheses, such the name of a system actor 
	(e.g. 
	$\alice$) or the 
	name 
	of an item (e.g. $\psw$ to refer to a password). Function symbols with non-zero 
	arity surround within parentheses their arguments, such as 
	$\knows(\alice,\psw)$ to denote that $\alice$ knows password
	$\psw$.

	\paragraph{The transition relation} The transition relation is defined through 
	\emph{transition rules}.
	Every transition rule is of the form $\sosrule{name}{C}{~s \xrightarrow{a} 
		f(s)~}$, and
	contains a condition $C$ and a conclusion $s \xrightarrow{a} f(s)$. The name of 
	a transition rule
	is given left of the line. The condition $C$ consists of zero or more 
	predicates $p$ that
	must be present in state $s$ to enable the transition rule. The conclusion $s 
	\xrightarrow{a} f(s)$
	describes the state change when the transition occurs, with $f$ a function 
	mapping old state $s$ to a new state $f(s)$. The arrow is labelled with the 
	event that describes the transition. The predicates may contain variables which 
	are implicitly universally quantified.

	\subsection{A network security example}
	
	Next, we illustrate how the introduced system model can be used to provide a 
	small 
	network security scenario with a formal specification.

	\paragraph{System description}
	The system example consists of a computer network where users log in to 
	password-protected servers. 
	We use $\users$ and $\passwords$ to denote 
	the set of users and passwords, respectively. Ownership of a password is 
	denoted by the predicate $\knows(\auser, \apassword)$, where $\auser \in 
	\users$ and $\apassword \in \passwords$. 
	The predicate 
	$\located(\auser, \amachine)$ denotes that the user $\auser \in \users$ has 
	logged into the server $\aserver \in \servers$.

	Let 
	$\accepts(\aserver, \apassword)$ be a predicate denoting that server $\aserver 
	\in \servers$ 
	accepts password $\apassword \in \passwords$. Then, a user $\auser$ can connect 
	to a server 
	$\aserver$ if the user knows a password that the server accepts. This is 
	expressed 
	by the following rule. 
	
	\[
	\sosrule{\loginuser}
	{
		\auser \in \users, \aserver \in \servers, \apassword \in \passwords, 
		\knows(\auser,\apassword),
		\accepts(\aserver, \apassword)
	}
	{\astate
		\xrightarrow{\loginuser(\auser,\aserver)}
		{\astate} \cup \{ \located(\auser,\aserver) \}
	}
	\]
	
	Up to here we have specified a system where regular users 
	can log into a server, provided they have the right credential. 
	The next step is to formalise possible actions that a malicious user can take. 
	We start specifying the 
	\exploiting\ rule, which allows a malicious user 
	to, via an exploit,
	create a 
	new credential on a server. 
	
	\[
	\sosrule{\exploiting}
	{\auser\in \users, \aserver \in \servers, \apassword 
		\in \passwords, \knows(\auser, \apassword)
	}
	{\astate
		\xrightarrow{\exploiting(\auser,\aserver)}
		{\astate} \cup \{\stores(\aserver, \apassword)\}
	}
	\]

	The 
	rule $\bruteforce$ specifies how an attacker can bruteforce a password and find 
	a 
	suitable credential on a server. 
	
	\[
	\sosrule{\bruteforce}
	{\auser \in \users, \apassword \in \passwords,
		\stores(\aserver,\apassword)
	}
	{\astate
		\xrightarrow{\bruteforce(\auser,\aserver)}
		{\astate} \cup \{ \knows(\auser,\apassword) \}
	}
	\]
	
	Lastly, the  
	rule $\eavesdrop$ expresses that a malicious user can eavesdrop on a connection 
	of 
	another user in order to find a credential 
	for 
	the server.
	
	\[
	\sosrule{\eavesdrop}
	{\auser, \auser' \in \users, \aserver \in \servers, \apassword \in \passwords, 
		\located(\auser, \aserver),\\
		\knows(\auser,\apassword),
		\stores(\aserver,\apassword)
	}
	{\astate
		\xrightarrow{\eavesdrop(\auser', \auser, \aserver, \apassword)}
		{\astate} \cup \{
		\knows(\auser',\apassword)\}
	}
	\]
	
	In this rule, $\auser'$ is the malicious user.

	\noindent \emph{Initial state.}
	The initial state of our system consists of a user $\alice$ with password 
	$\apassword_{\alice}$ and a server $\aserver$ that accepts passwords 
	$\apassword_{\alice}$. We consider another user $\mallory$ with password 
	$\apassword_{\mallory}$ that is not accepted by $\aserver$. That is to say, 
	our 
	initial state $s_0$ is $\{\knows(\alice, \apassword_{\alice}), \knows(\mallory, 
	\apassword_{\mallory}), \accepts(\aserver, \apassword_{\alice})\}$. 
	
	\noindent \emph{Final states.}
	Lastly, we define the set of final states to be any state where $\mallory$ is 
	logged in the server, i.e. $\finalstates = \{s \in \states| 
	\located(\mallory, \server) \in s\}$. Each of those final states represents 
	a security breach, one where $\mallory$ logged in to the server without being 
	authorised.

	\noindent \emph{Paths.}
	By analysing the system transition rules, one can compute the following 
	three paths that lead to a successful attack, i.e. that allows the 
	system to move from the initial state to a final state. Note that, this is not 
	an exhaustive list of paths, but a list we use for illustration purposes.

	\begin{enumerate}
		\item 
		$s_0$ $\xrightarrow{\exploiting(\mallory,\aserver)}$ 
		$ s_0 \cup \{\stores(\aserver, \apassword_{\mallory})\}$ \\
		$\xrightarrow{\loginuser(\mallory,\aserver)}$ 
		$ s_0 \cup \{\stores(\aserver, \apassword_{\mallory}), 
		\located(\mallory,\aserver)\}$
		\item $s_0$ $\xrightarrow{\bruteforce(\mallory,\aserver)}$ 
		$s_0 \cup \{ \knows(\mallory,\apassword_{\alice}) \}$ \\
		$\xrightarrow{\loginuser(\mallory,\aserver)}$ 
		$s_0 \cup 
		\{\knows(\mallory,\apassword_{\alice}),\located(\mallory,\aserver)\}$
		\item $s_0$ $\xrightarrow{\loginuser(\alice,\aserver)}$ 
		$s_0 \cup \{ \located(\alice,\aserver) \}$ \\
		$\xrightarrow{\eavesdrop(\mallory, \alice, \aserver)}$ 
		$s_0 \cup 
		\{\located(\alice,\aserver),\knows(\mallory,\apassword_{\alice})\}$ \\
		$\xrightarrow{\loginuser(\mallory,\aserver)}$ 
		$s_0 \cup 
		\{\located(\alice,\aserver),\knows(\mallory,\apassword_{\alice}),
		\located(\mallory,\aserver)\}$
	\end{enumerate}

	\subsection{Attacks and their goals}
	
	The output expected from a system specification is a list of attacks and a goal relation. We shall define an \emph{attack} as a sequence of \emph{attack goals}. Because an execution of the system is a sequence of transitions of the type $s
	\xrightarrow{a} s'$ where the action $a$ leads the system into the state $s'$ by removing and adding system properties from $s$, we consider the \emph{goal} of such transition to be exactly those properties that have been added/removed from $s$. Likewise, the goal of a sequence of transitions $s_i 
	\xrightarrow{a_i} s_{i+1} \cdots s_{j-1} 
	\xrightarrow{a_{j-1}} s_j$ consists of adding to and removing from $s_i$ properties in such a way that the system reaches state $s_j$. 
	This leads to the following definition of attacks, goals, and their relation.

	\begin{definition}[Goal relation]
		An \emph{attack goal} is a pair of sets of predicates $(P^{-}, P^{+}) \in 
		\powerset{\predset} 
		\times \powerset{\predset}$ where 
		$P^{-}$ and $P^{+}$ contains the system properties the attacker 
		wishes to remove and add, respectively. We say that an SP graph  $G = 
		(P^{-}_i, 
		P^{+}_i) 
		\seqop 
		\ldots 
		\seqop 
		(P^{-}_j, P^{+}_j)$ obtained from an LTS subpath, 
		$s_i 
		\xrightarrow{a_i} s_{i+1} \cdots s_{j-1} 
		\xrightarrow{a_{j-1}} s_j$, \emph{satisfies} a goal $(P^{-}, 
		P^{+})$, denoted $G \sat (P^{-}, 
		P^{+})$, if 
		the following two conditions 
		hold: 
		\begin{enumerate}
			\item $P^{-} \subseteq 
			s_i \setminus 
			s_j$, meaning that each property in $P^{-}$ has been removed from the 
			system. 
			\item $P^{+} \subseteq s_j \setminus 
			s_i$, meaning
			that every 
			property in $P^{+}$ has been added 
			to the system. 
		\end{enumerate}

	\end{definition}

	Take the following transition as an example.
	
	\begin{align*}
	& \{\knows(\alice, \apassword_{\alice}), 
	\knows(\mallory, 
	\apassword_{\mallory}), \accepts(\aserver, \apassword_{\alice})\} 
	\xrightarrow{\exploiting(\mallory,\aserver)} \\
	& \{\knows(\alice, \apassword_{\alice}), \knows(\mallory, 
	\apassword_{\mallory}), \accepts(\aserver, \apassword_{\alice})	
	\stores(\aserver, \apassword_{\mallory})\}
	\end{align*}
	
	Rather than carrying all the information in the transition, we propose to 
	succinctly express this step by $(\emptyset, \{\stores(\aserver, 
	\apassword_{\mallory})\})$, which describes 
	an attack step 
	where the adversary manages to add the password $\apassword_{\mallory}$ to the 
	server $\aserver$ by launching an exploit. 
	We extend this transformation to subpaths in the 
	straightforward way, which would give the following SP graphs 
	corresponding to the three LTS paths 
	obtained from our network example: 
	
	\begin{enumerate}
		\item 
		$ (\emptyset, \{\stores(\aserver, \apassword_{\mallory})\}) \seqop$ 
		$ (\emptyset, \{\located(\mallory,\aserver)\})$
		\item 
		$(\emptyset, \{ \knows(\mallory,\apassword_{\alice}) \}) \seqop$ 
		$(\emptyset, 
		\{\located(\mallory,\aserver)\})$
		\item 
		$(\emptyset, \{ \located(\alice,\aserver) \}) \seqop$  
		$(\emptyset, 
		\{\knows(\mallory,\apassword_{\alice})\}) \seqop$  
		$(\emptyset, 
		\{\located(\mallory,\aserver)\})$
	\end{enumerate}


	The generation algorithm requires attacks expressed as SP graphs and a goal 
	relation. 
	A straightforward transformation of an LTS path 
	into an SP graph is as follows. 
	A single transition $s
	\xrightarrow{a} s'$ results in the simple SP graph $(s, a, s')$, i.e. in an SP 
	graph with a single vertex labelled $(s, a, s')$. A subpath 
	$s_i 
	\xrightarrow{a_i} s_{i+1} \cdots s_{j-1} 
	\xrightarrow{a_{j-1}} s_j$ 
	is transformed into the SP graph 
	$(s_i, a_i, s_{i+1}) \seqop \ldots \seqop (s_{j-1}, a_{j-1}, s_j)$ where 
	$\seqop$ is the sequential composition operator. 
	In this case, an SP graph label carries the same information as a 
	transition in an LTS path. 
	
	Because SP graph labels become attack tree labels in our attack-tree generation 
	approach, we argue for a transformation that produces smaller labels. In 
	particular, 
	we propose a transformation where each SP graph label focuses on what the 
	attacker 
	has accomplished by making a transition from a state $s$ to another 
	state $s'$, rather than on the entire system state. Given a transition 
	$s 
	\xrightarrow{a} 
	s'$, what the attacker has accomplished is the subtraction of the properties
	$s 
	\setminus s'$ from the system, where $\setminus$ is the set subtraction 
	operator, and the addition of the properties $s' \setminus s$. 
	Therefore, 
	we transform a transition $s 
	\xrightarrow{a} 
	s'$ into a simple SP graph 
	$(s \setminus s', s' \setminus s)$.

	Note that an SP graph over $\powerset{\predset} 
	\times \powerset{\predset}$ may satisfy many goals. In particular, if 
	$G \sat 
	(X, 
	Y)$, then $G \sat (X', 
	Y')$ for every $X' \subseteq X, Y' \subseteq Y$. That is, the fewer predicates 
	in a goal, the more \emph{achievable} it becomes. 
	In particular, $G \sat (\emptyset, \emptyset)$ 
	for every SP graph $G$, which makes $\sat$ a total relation.	
	Further note that none of the SP graphs that 
	we produce from an LTS specification contains 
	parallel composition. This means that attack trees whose semantics is 
	$\attacksfunc{\anlts}$ may contain disjunctive ($\OR$) and 
	sequential conjunctive ($\SAND$) refinements, but not conjunctive 
	refinement ($\AND$). 
	
	We are ready now to establish the main result of this section, which shows that 
	the goal relation above allows for an efficient implementation of the 
	$\textsc{FindOptimalLabel}$. 
	
	\begin{theorem}\label{theo-findlabel}
		Let $\spgraphset = \{g_1, \ldots, g_n\}$ be a set of SP graphs obtained, 
		respectively, from the following LTS subpaths, 
		\begin{align*}
			& \atrace_1 = s_{i_1}^1 
			\xrightarrow{a_{i_1}^1} s_{i_1+1}^1 \cdots s_{j_1-1}^1 
			\xrightarrow{a_{j_1-1}^1} s_{j_1}^1 \\
			& \cdots \\	
			& \atrace_n = s_{i_n}^n 
			\xrightarrow{a_{i_n}^n} s_{i_n+1}^n \cdots s_{j_n-1}^n 
			\xrightarrow{a_{j_n-1}^1} s_{j_n}^n \\
		\end{align*} 
		A goal $b$ satisfying that $\spgraphset \subseteq \attacksfunc{b}$ 
		with $|\attacksfunc{b}|$ minimum, i.e. a solution to the 
		\textsc{FindOptimalLabel} problem on input $(\sat, \spgraphset)$, can be 
		calculated by, 
		\begin{align*}
			b = \left( (s_{i_1}^1 \setminus s_{j_1}^1) \cap \ldots \cap (s_{i_n}^n 
			\setminus s_{j_n}^n),  (s_{j_1}^1 \setminus s_{i_1}^1) \cap \ldots 
			\cap (s_{j_n}^n 
			\setminus s_{i_n}^n) \right)		
		\end{align*}
	\end{theorem}
	\begin{proof}
		Recall that $\attacksfunc{b}$ contains all SP graphs that satisfy $b$. 
		The proof is based on the following property of the goal relation: if $g 
		\sat 
		(X, 
		Y)$,  $X' \subseteq X$ and $Y' \subseteq Y$, then $g \sat (X', 
		Y')$. A corollary of this result is that $\attacksfunc{(X, Y)} \subseteq 
		\attacksfunc{(X', Y')}$.

		Assume there exists another 
		common goal $b' = (P^{-}, P^{+})$ with $b' <_{\spgraphset} b$. Because $b'$ is satisfied by all SP graphs in 
		$\spgraphset$, it must be 
		the case that $P^{-} \subseteq 
		(s_{i_k}^k \setminus s_{j_k}^k)$ and  $P^{+} \subseteq 
		(s_{j_k}^k \setminus s_{i_k}^k)$ for every $k \in \{1, \ldots, n\}$, which 
		gives $P^{-} \subseteq (s_{i_1}^1 \setminus s_{j_1}^1) \cap \ldots 
		\cap (s_{i_n}^n 
		\setminus s_{j_n}^n)$ and $P^{+} \subseteq (s_{j_1}^1 \setminus s_{i_1}^1) 
		\cap \ldots 
		\cap (s_{j_n}^n 
		\setminus s_{i_n}^n)$. This gives that 
		$\attacksfunc{b} \subseteq \attacksfunc{b'}$, which gives that    
		$\attacksfunc{b} \setminus \spgraphset \subset  
		\attacksfunc{b'} \setminus \spgraphset$, which contradicts $b' <_{\spgraphset} b$. 
	\end{proof}
	
	The $\textsc{FindOptimalLabel}$ procedure derived in 
	Theorem \ref{theo-findlabel} calculates $2n$ set intersections. In general, 
	that would result in a worst-cased complexity $\mathcal{O}(n\cdot 
	|\predset|)$. In practice, however, we argue this algorithm is much more 
	efficient as labels only contain changes in the 
	system states, hence they should be relatively small in comparison to the 
	system states. 
	
	We finally have all the ingredients necessary to showcase our attack-tree 
	generation 
	approach by using the output of the LTS specification from our running example on the input of the three SP 
	graphs listed above. 
	The result is displayed by Figure 
	\ref{fig-generated-tree-example}.  When drawing the tree we 
	removed unnecessary 
	information to make the labels easier to read. We also added to the leaf nodes' 
	labels 
	the name of the action of the transition used to obtain the leaf node's label. 
	Even though actions are not part of an attack goal, they can be carried as 
	metadata to improve the readability of the tree.

	\begin{figure}
		\centering
		\includegraphics[width=1.2\textwidth, 
		angle=0]{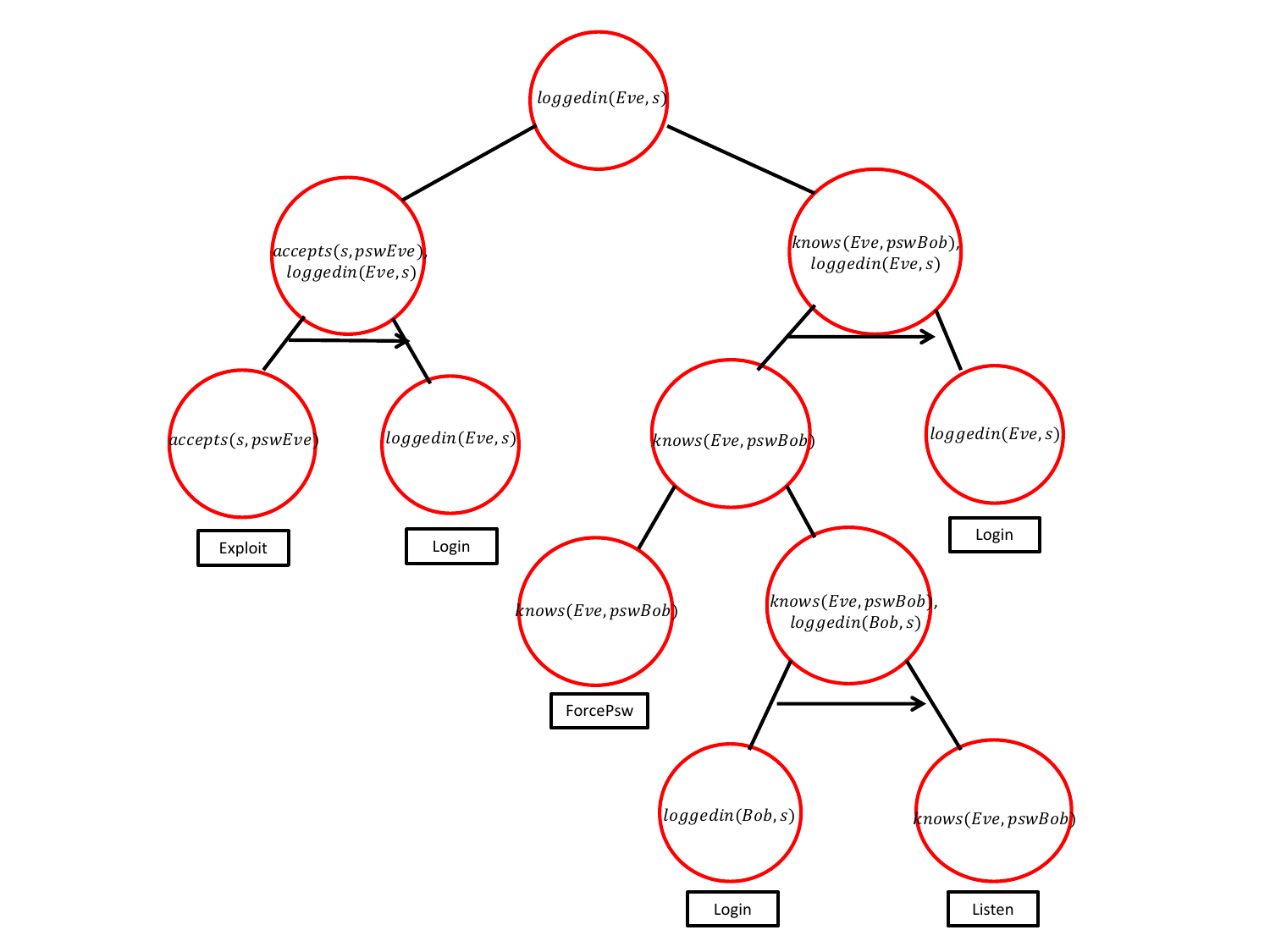}
		\caption{A graphical representation of the tree generated by our approach 
			on input the three SP graphs produced from our network example. 
			\label{fig-generated-tree-example}
		}
	\end{figure}	
	
	The resulting tree expresses the following attack scenario. An attack occurs 
	when a malicious user $\mallory$ logged into the server. $\mallory$ can do so 
	in two ways. On the one hand, $\mallory$ can include her own password in the 
	server, which she accomplishes via an exploit. On the other hand, $\mallory$ 
	may 
	learn $\alice$'s 
	password, which she accomplishes in two ways: by brute-forcing $\alice$'s 
	password or eavesdropping 
	on the communication 
	between $\alice$ and the server. The latter requires $\alice$ to log in to the 
	server. 
	
	Notice how the produced tree provides \emph{just} the necessary information 
	needed to understand the attacker's goal at each level of the tree. For 
	example, the 
	root's children on the right says that $\mallory$ 
	needs to learn $\alice$'s password, but it does not tell how. The deeper the 
	analyst goes into the tree, the more details about the attack they obtain.

\section{Conclusions}\label{sec-conclusion}
We have introduced 
the first attack-tree generation framework
that optimises the labelling and shape of the generated trees, while 
guaranteeing their soundness and correctness with respect to a system 
specification. 
The core of the 
article  
is dedicated to describing and proving correct the various algorithms involved 
in the generation approach. Particularly relevant is our reduction of the tree 
minimisation problem to 
the factorisation problem of algebraic 
expressions, which calls for further research on the factorisation of 
(commutative) semirings without multiplicative 
identity. 
We showed the feasibility of our approach by introducing a system model that 
couples well with our generation framework, and validated the entire approach 
via a small running example. 
Future work will be directed towards 
interfacing the generation framework with a model checker tool, such as 
mCRL2\footnote{https://www.mcrl2.org/web/index.html}, to allow for a 
large-scale validation of the approach. Such validation may include a study of whether the objective functions of our attack-tree generation problem, namely minimisation of size and information loss, do lead to trees that are easier to read and comprehend.  
We also plan to extend our approach to include defenses too, resulting into an attack-defense tree model.

	\appendix
	
	\section*{Appendix}
	\def\block{\par\noindent{\bf Theorem \ref{theo-cover-reduction}.\ } 
\ignorespaces}
\def\endtheorem{}

\begin{block}
	The set 
	cover problem is polynomial-time reducible to the 
	problem of finding the smallest partition $\{\spgraphset_1, \ldots, 
	\spgraphset_{k}\}$ of a set of graphs $\spgraphset$ such that, for every $i 
	\in 
	\{1, \ldots, 
	k\}$, 
	$\spgraphset_i$ 
	is an homogeneous set of SP graphs and there exists a label 
	$\abasicact$ such that $\forall g \in \spgraphset_i \colon g \sat 
	\abasicact$.
\end{block}
\begin{proof}
	First, we describe the set cover problem. 
	Let $U$ be a set 
	of 
	elements and 
	$\mathcal{S} = \{S_1, 
	\ldots, 
	S_m\}$ subsets 
	of $U$. A set $\mathcal{C} \subseteq \mathcal{S}$ is said to be a cover if 
	$\bigcup_{S \in \mathcal{C}}S 
	= U$. The set cover problem consists in finding a cover $\mathcal{C}$
	of minimum 
	cardinality. 
	
	Without loss of generality, assume $U = \{g_1, \ldots, g_n\}$ is a 
	homogeneous set of SP graphs. Likewise, we shall consider 
	$\{S_1, 
	\ldots, 
	S_m\}$ to be labels of \SAND{} trees. 
	Build the goal relation $\sat$ as follows:
	\[
	g_i \sat S_j \iff g_i \in S_j
	\]

	First, let us show how to obtain, given a cover 
	$S_{i_1}, \ldots, S_{i_{k}}$, a partition $\{\spgraphset_1, \ldots, 
	\spgraphset_{k}\}$ of $U$ such that, for every $i 
	\in 
	\{1, \ldots, 
	k\}$, 
	there exists a label 
	$S_j$ such that $\forall g \in \spgraphset_i \colon g \sat 
	S_j$. The algorithm works as follows.
	From $j = 1$ to $j = k-1$, let $\spgraphset_j$ be the set containing all 
	elements in $S_{i_j}$ that do not appear in the remaining sets 
	$S_{i_{j+1}}, 
	\ldots, S_{i_{k}}$. Let $\spgraphset_k = S_{i_k}$. 
	It follows that $\{\spgraphset_1, \ldots, 
	\spgraphset_{k}\}$ is a partition of $U$. Now, because $\spgraphset_j 
	\subseteq S_{i_j}$, we obtain that 
	$\forall g \in \spgraphset_j \colon g \sat 
	S_{i_j}$. 
	

	Second, let us show how to obtain a solution to the set cover 
	problem,  
	given a 
	partition $\{\spgraphset_1, \ldots, 
	\spgraphset_{k}\}$ of minimum cardinality 
	of $U$ such that, for every $i 
	\in 
	\{1, \ldots, 
	k\}$, 
	there exists a label 
	$S_j$ such that $\forall g \in \spgraphset_i \colon g \sat 
	S_j$. 
	Let 
	$f\colon \{1, \ldots, k\} \rightarrow \{1, \ldots, m\}$ be a total mapping 
	satisfying that,  for 
	every $i \in \{1, \ldots, k\}$, $g 
	\in \spgraphset_i \implies g \sat S_{f(i)}$. Note that there exists at 
	least one such mapping.
	Because $g \sat S_{f(i)} \implies g \in S_{f(i)} $, we obtain that 
	$\spgraphset_i \subseteq S_{f(i)}$. Therefore, $S_{f(1)}, \ldots, S_{f(k)}$ 
	is a covering of $U$. 
	
	We conclude the proof by showing that $S_{f(1)}, \ldots, S_{f(k)}$ is a 
	minimum 
	cover. 
	This follows from our first proof step, which shows that a cover $S_{i_1}, 
	\ldots, S_{i_{k'}}$ of $U$ leads to a partition of 
	cardinality $k'$. That is, the existence of a smaller cover implies the 
	existence of a smaller partition, which contradicts the 
	premise that $\{\spgraphset_1, 
	\ldots, \spgraphset_k\}$ has minimum 
	cardinality. 
	%
	It is worth remarking that this generalises the proof provided in our 
		previous 
		work \cite{GJMTW17}, as it imposes no restriction on the attacks, nor 
		on 
		the goal
		relation. 
	
\end{proof}

\end{document}